\documentclass[11pt,a4paper]{article}
\usepackage{amsmath, amssymb, amsthm, amsfonts, latexsym, graphicx, color, url,complexity}
\usepackage{authblk}

\newlength{\actualtopmargin}
\newlength{\actualsidemargin}
  \theoremstyle{plain}
  \newtheorem{theorem}{Theorem}
  
  \newtheorem{lemma}[theorem]{Lemma}
  \newtheorem{corollary}[theorem]{Corollary}

  \theoremstyle{definition}
  \newtheorem{definition}[theorem]{Definition}

  \theoremstyle{remark}
  
  \theoremstyle{plain}
  \newtheorem*{theorem*}{Theorem}
  \newtheorem*{lemma*}{Lemma}
  \newtheorem*{corollary*}{Corollary}
  \newtheorem*{proposition*}{Proposition}
  \newtheorem*{claim*}{Claim}
  \newtheorem*{problem*}{Problem}
  
\newcommand{\ii}{\mathbb{I}}

\renewcommand{\poly}{\textrm{poly}}

\newcommand{\prob}[1]{\Pr[ #1 ]}

\newcommand{\norm}[1]{\left\| #1 \right\|_2}
\newcommand{\opnorm}[1]{\left\| #1 \right\|}
\newcommand{\bra}[1]{\langle #1 \vert}

\newcommand{\ket}[1]{\vert #1 \rangle}

\newcommand{\braket}[2]{\langle #1 \vert #2 \rangle}

\newcommand{\half}{\frac{1}{2}}

\newcommand{\bigO}[1]{O\left(#1\right)}
\newcommand{\CC}{\mathbb{C}}
\newcommand{\RR}{\mathbb{R}}
\newcommand{\Herm}{\mathrm{Herm}}

\begin{document}
\title{\Large \textbf{
Shorter unentangled proofs for Ground State Connectivity.
}}
\author[1]{Libor Caha}
\affil[1,2]{Research Center for Quantum Information, Institute of Physics, Slovak Academy of Sciences, D\'ubravsk\'a cesta 9, 845 11 Bratislava, Slovakia}
\author[2]{Daniel Nagaj\footnote{daniel.nagaj@savba.sk}}
\author[3]{Martin Schwarz}
\affil[3]{Dahlem Center for Complex Quantum Systems, Freie Universität Berlin, 14195 Berlin, Germany}
\maketitle
\vspace{-5mm}


\begin{abstract}
	Can one considerably shorten a proof for a quantum problem by using a protocol with a constant number of unentangled provers?
		We consider a frustration-free variant of the $\QCMA$-complete Ground State Connectivity (GSCON) problem for a system of size $n$ with a proof of superlinear-size. 
		We show that we can shorten this proof in $\QMA(2)$: there exists a two-copy, unentangled proof with length of order $n$, up to logarithmic factors,
	while the completeness-soundness gap of the new protocol becomes a small inverse polynomial in $n$.
\end{abstract}

\section{Introduction: Unentangled Provers and Short Proofs}

While entanglement is essential for quantum algorithms, 
{\em unentanglement} can also be an interesting resource. 
In quantum complexity, such a guarantee
about a purported proof can significantly improve the power of a~verifier.
Blier and Tapp \cite{BlierTapp} discovered that two unentangled copies of a {\em short} witness of the type
\begin{align}
\frac{1}{\sqrt{n}}\sum_{i=1}^{n} \ket{i}\ket{c_i},
\label{color}
\end{align}
can be used to prove the existence of a solution for the NP-complete graph coloring problem. All one needs is $c_i$ to be the color of the vertex $i$ in the solution. Listing the color for each vertex would normally take space on the order of $n$, while
the two-copy, unentangled quantum proof takes space only $2+\log n$, as we need 2 qubits to encode the three possible colors.
We (the verifier) can check this proof as follows. First, 
let us measure $\ket{i}\ket{c_i}$ from one witness and $\ket{i'}\ket{c_{i'}}$ from the other witness. Sometimes, we get results for neighboring vertices $i,i'$, so we can check if $c_i \neq c_{i'}$, verifying the validity of the coloring. However, we also need to thwart cheating provers by a \textsc{swap} test \cite{QuantumFingerprinting} and a color-measuring test checking the consistency of the two copies of the witness and well defined vertex colors, and a another to make sure the superposition contains info about all vertices. Only when we are sure that the two witnesses are unentangled, these tests are sound, while an entangled state could easily fool the \textsc{swap} test.

The new quantum proofs are exponentially shorter, so one might think 
we could use a quantum computer to quickly find them (in \BQP). However, there is no straightforward way for this, e.g. using variants of Grover's search, as one needs to keep the proofs unentangled. Therefore, this result does not imply anything about the containment of $\NP$ in \BQP. On the other hand, it is connected to interesting questions about the the nonexistence of perfect disentanglers \cite{Beigi0810.5109}
or the strong-NP hardness of separability testing for density matrices \cite{Gurvits, PhDLiu, GharibianStrongNPHardness}.

The main price we pay for shortening the proof in \cite{BlierTapp} is that the completeness-soundness gap is small -- the probability of detecting cheating provers and thus the gap is $\Omega(n^{-6})$. However, there are also independent results that analyze the possible tradeoff between the proof length and the (completeness-soundness) gap.  The protocol of Aaronson et al. \cite{Beigi0810.5109} looks at the balanced 2-out-of-4-SAT problem, relies on Dinur's proof of the PCP theorem \cite{Dinur:2007:PTG:1236457.1236459}, and produces constant soundness and perfect completeness, while using $\tilde{O}\left(\sqrt{n}\right)$ unentangled copies of the proof. Also, instead of \eqref{color}, it uses a phase encoding $\ket{\psi} = \frac{1}{\sqrt{2^n}}\sum_{j} (-1)^{c_j}\ket{j}$ of the witness.
Next, Beigi \cite{Beigi2outof4SAT} also has a protocol for 2 provers sending $O(\log n)$ qubits, with gap $\Omega(n^{-3-\epsilon})$. Meanwhile, the product test of Harrow and Montanaro \cite{ProductTest} applied to \cite{Beigi0810.5109} has lead to a 2 prover protocol sending $\tilde{O}(\sqrt{n})$ qubits with a constant gap. 
Investigating unentanglement further, Chen and Drucker \cite{ChenDrucker} found a protocol for 2-out-of-4-SAT using unentangled measurements with $\tilde{O}(\sqrt{n})$ provers sending $O(\log n)$ qubits. Next, Le Gall, Nakagawa and Nishimura \cite{LGNN} gave an improved protocol for 3-SAT with only two log-size, unentangled quantum proofs and a $\Omega(1/n\, \textrm{polylog}(n))$ completeness-soundness gap.
Chiesa and Forbes \cite{cj13-01} provided a tighter soundness analysis leading to  $\Omega(n^{-2})$ completeness-soundness gap for \cite{BlierTapp} and a smooth trade-off between $\mathcal{K}$ provers and a gap $\Omega(\mathcal{K}^2 n^{-1})$ for \cite{ChenDrucker}.
A similar gap improvement for \cite{BlierTapp} was proved by Nishimura and Nakagawa in \cite{BlierTappSoundnessGapImprovement}.

These results mainly concern short proofs of {\em classical} problems. 
Inspired by them, we choose to look at a naturally {\em quantum} problem, Ground State Connectivity (GSCON), 
and ask whether we could rely on unentanglement to make its proof shorter.  
This is indeed what we find, for a particular $\QCMA$-complete variant of GSCON.
However, our result has two shortcomings. First, the shortening is significant only if the original proof is superlinear. Second, the completeness-soundness gap becomes very small. It should thus serve as a proof of principle that opens the door to other more effective unentanglement-based constructions of proof systems for quantum problems. 

We call for a general investigation of when and how much proofs for quantum complexity classes could be shortened, when relying on unentanglement. Note that the relationship of the class QMA(2) to classes without unenanglement is not fully understood yet. One of the things we know is that if the verifier could only perform one-way LOCC measurements on a constant number of unentangled proofs, his power would diminish, in particular $\QMA^{\text{LOCC}}_{\ell(n)}(2)_{c,s}\subseteq\QMA_{O(\ell^2(n)\epsilon^{-2}),c,s+\epsilon}$, as shown by Brand\~ao, Christandl and Yard \cite{BCY}. On the other hand, adding the unentanglement requirement doesn't allow one to freely shorten proofs of $\QMA$. Unless a subexponential-time quantum algorithm for 3-SAT exists, the size of a $\QMA$ witness cannot be shortened to less than its squareroot in $\QMA(2)$ with a constant completeness-soundness gap, i.e. $\QMA_{n}(2)\not\subseteq \QMA_{o(n^2)}$.

Let us now present our results. We start with a review of the GSCON problem in Section~\ref{sec:gscon},
and present a high-level view of our protocol and state the main theorem in Section~\ref{sec:shorter}. In Section~\ref{sec:tests} we give the details of the proof verification procedure, and prove our main result in Sections~\ref{sec:soundness} (soundness), \ref{sec:completeness} (completeness) and \ref{sec:csgap} (gap lower bound).


\section{The Ground space connectivity problem (GSCON)}
\label{sec:gscon}

Let us start  with the definition of the $\QCMA$-complete Ground State Connectivity (GSCON) problem \cite{gscon} about the possibility of traversal between two low-energy states for a local Hamiltonian, using local unitary transformations, while remaining in a low-energy sector.

\begin{definition}[The Ground State Connectivity (GSCCON) problem  
\cite{gscon}] 
\label{def:GSCON}
Ground state connectivity (GSCON) with parameters $H, n, k, R, \eta_1, \eta_2, \eta_3, \eta_4, \Delta, m, U_{\psi}, U_{\phi}$ is a promise problem defined as follows. Consider
\begin{enumerate}
\item a $k$-local Hamiltonian $H=\sum_i^R H_i$ acting on $n$ qubits with $R$ terms $H_i \in \Herm ((\CC^2)^{\otimes k})$ satisfying $||H_i||_{\infty} \leq 1$,
\item real numbers $\eta_1,\eta_2,\eta_3,\eta_4,\Delta\in \RR$, and an integer $m\geq0$, such that $\eta_2-\eta_1\geq\Delta$ and $\eta_4-\eta_3\geq\Delta$,
\item descriptions of polynomial size quantum circuits $U_{\psi}$ and $U_{\phi}$ generating the {\em starting} and {\em target} states $\ket{\psi}$
and $\ket{\phi}$ from the initial state $\ket{0}^{\otimes n}$, satisfying $\bra{\psi}H\ket{\psi}\leq\eta_1$ and $\bra{\phi}H\ket{\phi}\leq\eta_1$, respectively.
\end{enumerate}
Decide, which of the two cases is true:
\begin{enumerate}
\item[{\em YES}:] There exists a sequence of 1 and 2 qubit\footnote{In general, this could be also $l$-local unitaries, we choose $l=2$. This variant of the problem is still QCMA complete \cite{gscon}.} unitaries $\{U_i\}^m_{i=1}$ such that
\begin{enumerate}
\item intermediate states remain in low energy space, i.e. for all $i\in[m]$ and intermediate states $\ket{\psi_i} := U_i\ldots U_2U_1\ket{\psi}$, one has $\bra{\psi_i}H\ket{\psi_i}\leq \eta_1$, and
\item the final state is close to the target state, i.e. $\|U_m\ldots U_1\ket{\psi}-\ket{\phi}\|_2\leq\eta_3$.
\end{enumerate}
\item[{\em NO}:] For all 1 and 2 qubit sequences of unitaries $\{U_i\}^m_{i=1}$, either
\begin{enumerate}
\item some intermediate state has a high energy, i.e. there exists an $i\in[m]$, for which the intermediate state $\ket{\psi_i}:=U_i\ldots U_2U_1\ket{\psi}$, obeys $\bra{\psi_i}H\ket{\psi_i}\geq\eta_2$, or
\item the final state is far from the target state, i.e. $\|U_m\ldots U_1\ket{\psi}-\ket{\phi}\|_2\geq\eta_4$.
\end{enumerate}
\end{enumerate}
\end{definition}

In this paper, we consider a specific version that we call {\em frustration-free} GSCON. It requires an at least inverse-polynomial promise gap $\Delta=\Omega(1/\poly(n))$, and a positive semidefinite, frustration-free Hamiltonian, with $\eta_1=0$. 
We choose this for a technical reason, as we are presently unable to devise a strong enough low-energy testing procedure for the witnesses. However, this variant of GSCON is still $\QCMA$ complete.

We know that in general, GSCON (deciding whether a low-energy state $\ket{\psi}$ can be transformed to a low-energy state $\ket{\phi}$ using a sequence of $m=\poly(n)$ (2-)local gates, while remaining a low-energy state) is a QCMA complete problem.
The {\em frustration-free} GSCON variant still belongs to $\QCMA$, as the local transformations can be easily communicated classically, and their properties tested on a quantum computer. On the other hand, it is $\QCMA_1$ hard, as it also has instances that can be constructed (as in \cite{gscon}) for a Hamiltonian related to the verification procedure for a $\QCMA_1$ proof -- with perfect completeness. However, thanks to $\QCMA=\QCMA_1$ \cite{JordanNagajNishimura}, this must also be $\QCMA$ hard. Therefore, frustration-free GSCON is also $\QCMA$ complete.

We assume the circuits $U_\psi$ and $U_\phi$ are given in terms of 1 and 2-qubit unitary gates. All input parameters are specified with rational entries, each using $O(\poly(n))$ bits of precision.
We expect the same for the gates $U_i$ that are chosen out of $G =\textrm{poly}(n)$ possible gates (including the target qubit specification), encoded as bit strings of length at most $O(\log n)$, with polynomial-precision entries.

The standard proof for GSCON is the list of unitary transformations that generate the low-energy states traversing from $\ket{\psi}$ to $\ket{\phi}$. In the next Section, we devise a different type of proof involving superpositions.


\section{Shorter proofs for Ground State Connectivity relying on unentanglement.}
\label{sec:shorter}


\subsection{A shorter proof: the sequence of states in superposition}
\label{sec:superpos}

The original proof has size $m \log G = O(m \log(n))$, as it holds the information about the $m$ gates $U_i$ applied to the initial state (each $U_i$ is a 1 or 2 qubit unitary gate chosen from a set of size $G$, including the target qubits specification).
We want to shorten it to 
\begin{align}
	(\log m) \times \left(n+\log G\right) = O(n \log n) ,
\end{align}
at the cost of a smaller completeness-soundness gap, and asking for four unentangled proofs. Later in Corollary~\ref{cor:qma2} we show that only two unentangled proofs suffice.

\begin{figure}[!h]
\begin{center}
	\includegraphics[width=8cm]{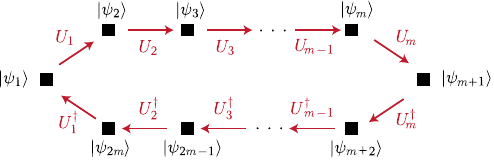}
	\caption{A cycle of states $\ket{\psi_1}$, \dots, $\ket{\psi_{m+1}}$, \dots, $\ket{\psi_{2m}}$, connected via the unitaries $U_1, \dots, U_m, U_m^\dagger,\dots, U_1^\dagger$.
	\label{fig:cycle}
}
\end{center}
\end{figure} 

We ask for {\em two unentangled copies} of the 
two-register ({\em label} and {\em gate})
state
\begin{align}
\ket{U} =\frac{1}{\sqrt{2m}}\sum_{i=1}^{2m} \ket{i} \ket{u_i},
\label{GSCONproofU}
\end{align}
encoding a cycle of local transformations as in Figure~\ref{fig:cycle}, with
each $u_i$ a classical string decribing the gate $U_i$ (chosen from a gate set of size $G = \textrm{poly}(n)$, including which qubits it acts on). 

We also ask for {\em two unentangled copies} of the 
two-register ({\em label} and {\em data})
state
\begin{align}
	\ket{S} =\frac{1}{\sqrt{2m}}\sum_{i=1}^{2m} \ket{i} \ket{\psi_i},
 \label{GSCONproof}
\end{align}
encoding a cyclical sequence of labeled low-energy states $\ket{\psi_i}$,
illustrated in Figure~\ref{fig:cycle}.
The sequence should start with the initial state $\ket{1}\ket{\psi_1} = \ket{1}\ket{\psi}$ for $\ket{\psi}$ from the definition of GSCON, and obey $U_i \ket{\psi_{i}} = \ket{\psi_{i+1}}$, with $U_{2m} \ket{\psi_{2m}} = \ket{\psi_{1}}$ at the end. 
The first half of the sequence corresponds to the traversal from $\ket{\psi}$ to $\ket{\phi}$ using the gates $U_i$. The second half should be its inverse, with $U_{m+i}= U_{m+1-i}^\dagger$, so that $U_{2m} \dots U_1 = \ii$. 

Observe that such a state $\ket{S}$ is invariant under the action of the unitary 
\begin{align}
	W = \sum_{i=1}^{2m} \ket{i+1}\bra{i}\otimes U_i, \label{Wunitary}
\end{align}
where we identify $\ket{2m+1} \equiv \ket{1}$ in the first register,
and assume $U_{m+i}= U_{m+1-i}^\dagger$ for $i=1,\dots,m$.

\subsection{The main result}
\label{sec:theorem}

Our main, superlinear proof-shortening result for frustration-free GSCON is the following Theorem:
 
\begin{theorem}[Shorter proofs for ff-GSCON in QMA($4$)]
\label{th:main}
Consider an instance of {\em Frustration-free GSCON} (ff-GSCON)
combining 
Definition~\ref{def:GSCON} with the extra assumptions of a positive-semidefinite, frustration-free Hamiltonian acting on $n$ qubits, 
with parameter $\eta_1 =0$ and an inverse-polynomial promise gap $\Delta$.
This promise problem has a proof system in $\QMA(4)$, with four unentangled proofs of length $\bigO{n \log n}$,
and an inverse polynomial\footnote{This inverse polynomial is quite small, as shown in Section~\ref{sec:csgap}: $c'-s' = \Omega\left(
\Delta^{13} m^{-32} G^{-10}
\right)$, with $\Delta$
from the definition of GSCON and $G$ the gate set size.} completeness-soundness gap. 
\end{theorem}

We present the protocol in Section \ref{sec:tests} and analyze it in detail in Sections \ref{sec:soundness} and \ref{sec:completeness}, proving Theorem~\ref{th:main}. Let us now show how to use this 4 unentangled witness protocol as a black box to build a procedure with only 2 witnesses, putting frustration-free GSCON into $\QMA(2)$ with shortened proofs.

\begin{corollary}
	Ff-GSCON is in $\QMA_{\bigO{n\log n}}(2)$ with an inverse polynomial completeness-soundness gap.
	\label{cor:qma2}
\end{corollary}

\begin{proof}
Our protocol from Section~\ref{sec:tests} uses 4 unentangled witnesses -- two copies of the state $\ket{U}$ and two copies of the state $\ket{S}$. We know how to use the \textit{$\QMA(k)$ to $\QMA(2)$ transformation} \cite{ProductTest} to place it in $\QMA(2)$ with the same asymptotic witness length and altered completeness and soundness. The new $\QMA(2)$ protocol asks for two identical witnesses -- in our case two copies of the state $\ket{U}\otimes\ket{U'}\otimes\ket{S}\otimes\ket{S'}$.
The verifier performs two tests with the same probability: a) The \textsc{product} test, or the b) the original $\QMA(k)$ protocol on one of the states. In \cite{ProductTest}, the authors showed the containment $\QMA_w(k)_{c',s'}\subseteq\QMA_{kw}(2)_{c'',s''}$, with completeness $c''=\frac{1+c'}{2}$, soundness $s''=1-\frac{(1-s')^2}{100}$, and new witness size $kw$. However, for the resulting completeness-soundness gap to be positive, there is a requirement on the original completeness and soundness, which  
our $\QMA(4)$ protocol might not fulfill.

However, this is not a problem. The trick is to use the $\QMA(k)$ to $\QMA(2)$ conversion with variable probabilities to run the tests a) and b). 
Let us label $c'$ the completeness and $s'$ the soundness of Test b), the $\QMA(k)$ protocol, and denote $p$  the probability to run Test a) and $1-p$ the probability to run Test b). 
Following the proof of Lemma 5 \cite{ProductTest}, we find that the resulting  $\QMA(2)$ protocol has completeness and soundness:
\begin{align}
	c''&=p+(1-p) c', \\
	s''&\leq\max_{\epsilon\leq \frac{512}{11}\delta_{\textrm{P}}}
		\left\{ p\left(1-\delta_{\textrm{P}}\right)
			+ (1-p)\min\{1,s'+\sqrt{\epsilon}\}\right\}, \label{sdd}
\end{align}
with 
$\epsilon$ a bound on how far the witness is from a product state,
and $\delta_{\textrm{P}}\geq\frac{11}{512}\epsilon$
the probability that the \textsc{product} test rejects it.
The maximum in \eqref{sdd} is achieved for $\sqrt{\epsilon}=1-s'=\sqrt{\frac{512}{11}\delta_{\textrm{P}}}$. Therefore, $s'' \leq 1-p\delta_{\textrm{P}} \leq 1-p(1-s')^2\frac{11}{512}$.
With this in hand, we realize that we can always tune $p$ to create a protocol with a positive, inverse-polynomial completeness-soundness gap. 
For example, we can achieve $c''-s''\geq 
\frac{11}{512}p(1-s')^2 - (1-p)(1-c')
\geq \frac{1}{50}(c'-s')^2$, by choosing
\begin{align}
	p=\frac{1-c'+ \frac{1}{50}(c'-s')^2}{1-c'+\frac{11}{512}(1-s')^2 }.
\end{align}
Observe that $0\leq p \leq 1$, as $c'-s' \leq 1-s'$.

Therefore, there exists a way to tune the probability $p$ for running the \textsc{product} test vs. the $\QMA(4)$-based composite procedure from Section~\ref{sec:tests},
giving us a $\QMA(2)$ protocol for ff-GSCON, with shorter proofs of size $O(n \log n)$, and a completeness-soundness gap inverse polynomial in $n$.
\end{proof}

\section{Proof of Theorem~\ref{th:main}.}

The proof of Theorem~\ref{th:main} is spread over four Sections. We first describe the proof system in Section~\ref{sec:tests}, show its soundness in Section~\ref{sec:soundness} and completeness in Section~\ref{sec:completeness},
and prove that the completeness-soundness gap is an inverse polynomial in $n$ in Section~\ref{sec:csgap}.

\subsection{The verification procedure}
\label{sec:tests}

Let us start the proof of Theorem~\ref{th:main} with the tests that we must run on the 4 unentangled proofs for GSCON.
Note that in Corollary~\ref{cor:qma2} we have shown
how to get away with only 2 unentangled witness states instead of 4, relying on an argument similar to the \textsc{product} test of Harrow and Montanaro \cite{ProductTest}, 
while decreasing the completeness-soundness gap (but still to an inverse polynomial in $n$.

The verifier asks the provers to provide two unentangled copies of the states $\ket{U}$ \eqref{GSCONproofU} and $\ket{S}$ \eqref{GSCONproof}, as described in Section~\ref{sec:superpos}. From now on, let us call these $\ket{U},\ket{U'},\ket{S},\ket{S'}$.
With probabilities 
\begin{align}
	p_i = \frac{r_i^{-1}}{\sum_{j}r_j^{-1}}, \qquad i=1,\dots,8, \label{pchoice}
\end{align}
where $r_i$ are listed in Figure~\ref{tab:prob}, 
the verifier randomly chooses to do perform one of the following set of eight tests, accepting if the test succeeds. 
We choose the threshold parameters $r_i$ and test probabilities $p_i$ in such a way that in the {\em NO} case of the ff-GSCON instance, it must be true at least one of the tests rejects with probability more than its $r_i$, so the verifier accepts the proof with probability at most 
\begin{align}
	s' \leq 1- r_i p_i = 1- \frac{1}{\sum_j r_j^{-1}}, \label{sfromr}
\end{align}
independent of $i$. 
On the other hand, in the {\em YES} case, we will show that this results in completeness $c'$ that is at least an inverse polynomial in $n$ above $s'$, as stated in Theorem~\ref{th:main}. Here are the tests:

\begin{enumerate}
\item (\textsc{swap~U}) Do a \textsc{swap} test on the unitary-encoding witnesses
$\ket{U}$ and $\ket{U'}$ and reject on failure. 
This test checks basic consistency between the witnesses.

\item (\textsc{unique}) Measure the label and gate register of the states $\ket{U}$
and $\ket{U'}$ in the computational basis. 
If the labels don't match, accept. On the other hand, if you obtain the same label from both copies, check if the gate register measurement results match. Reject if they don't. 
Also reject if the results do not encode unitaries from the expected gate set.
This test checks if the unitaries $U_i$ are well defined.

\item (\textsc{uniform}) Do a projective measurement on the gate register of the state $\ket{U}$ and accept if the result is not the uniform superposition $\ket{\bar{g}}$ of all possible gate-encoding states. Proceed otherwise
and measure the label register. Reject if the result is not
the uniform superposition $\ket{\bar{0}}$. Together with \textsc{swap} and \textsc{unique}, this test checks if the terms $\ket{i}\ket{u_i}$ for various $i$ are nearly uniformly present in $\ket{U}$.

\item (\textsc{swap~S}) Do a \textsc{swap} test on the states $\ket{S}$ and $\ket{S'}$ (the state sequence-encoding witnesses) and reject on failure.
This test checks basic consistency between the witnesses. 

\item (\textsc{sequence}) 
First, apply the unitary $W$ from \eqref{Wunitary} to $\ket{S}$ in a probabilistic fashion, consuming the state $\ket{U}$ in the process. Second, compare $W\ket{S}$ and $\ket{S'}$ using a \textsc{swap} test, rejecting on failure. In detail,

\begin{enumerate}
\item Combine $\ket{U}$ and $\ket{S}$, and apply the encoded unitaries (we assume the states $\ket{u_i}$ are computational basis states) from the gate register to the data register to form the state
$	\sum_{i} \frac{1}{\sqrt{2m}}\ket{i}\ket{u_i} \sum_j  \frac{1}{\sqrt{2m}} \ket{j} U_i\ket{\psi_j}$.
\item Project the gate register onto the uniform superposition state $\ket{\bar{g}}$.
Accept if the projection fails, and proceed otherwise. 
\item
At this point, we expect to work with the renormalized state 
$\sum_{i}  \frac{1}{\sqrt{2m}} \ket{i}\ket{\bar{g}} \sum_{j} \frac{1}{\sqrt{2m}} \ket{j} U_i \ket{\psi_j}. \label{Tgate}
$
 Drop the gate-encoding register with the state $\ket{\bar{g}}$.
\item Project onto identical label registers.
Accept if the projection fails, and continue otherwise. 
\item At this point, we expect to work with the renormalized state 
$\sum_i \frac{1}{\sqrt{2m}}\ket{i} \ket{i} U_i \ket{\psi_i}$. Uncompute and drop the second label register. 
\item Shift the label register by 1 in a cyclical fashion, with $2m$ becoming $1$, to obtain the state $\ket{T'}$.

\item 
Do a \textsc{swap} test between $\ket{T'}$ and $\ket{S'}$ and {\em reject} on failure. 
Note that for honest provers we expect $\ket{T'}=\frac{1}{\sqrt{2m}}\sum_{i} \ket{i+1}U_i\ket{\psi_i}$, identical to the state $\ket{S'}$.
\end{enumerate}

This test checks if all states $\ket{i}\ket{\psi_i}$ in $\ket{S}$ have significant amplitudes, and whether $\ket{\psi_{i+1}} = U_i \ket{\psi_i}$.

\item (\textsc{start}) Check if the sequence in $\ket{S}$ starts with the state $\ket{\psi}$ from the problem instance as follows:
\begin{enumerate}
	\item Measure the label register of $\ket{S}$. Accept if it is not $1$, otherwise continue.
	\item If the label is $1$, use another register to prepare $\ket{\psi}$, according to the problem instance.
	\item Perform a \textsc{swap} test between the data register of $\ket{S}$ and the prepared state.
	Reject on failure.
\end{enumerate}

\item (\textsc{end}) Check if the sequence in $\ket{S}$ ends near the state $\ket{\phi}$ from the problem instance as follows:
\begin{enumerate}
	\item Measure the label register of $\ket{S}$. Accept if it is not $m+1$, otherwise continue.
	\item In another register, prepare $\ket{\phi}$, according to the problem instance.
	\item Perform a \textsc{swap} test between the data register of $\ket{S}$ and the prepared state.
	Reject on failure.
\end{enumerate}

\item (\textsc{low}) 
Measure the label register of $\ket{S}$, and then
the energy of its data register.
Reject, if 
the energy is higher than $\frac{\eta_2}{2}$
with $\eta_2$ 
from the definition of GSCON.

This test checks if the traversed sequence of states is made only from low-energy states. 
Note that a weakness of this test makes us talk about ff-GSCON, i.e. a GSCON instance with $\eta_1=0$, involving a frustration-free, positive semidefinite Hamiltonian.
\end{enumerate}

Choosing one of the tests at random gives us a reasonable assurance that the state $\ket{U}$ contains a nearly uniform superposition of the sequence of labeled, computational-basis encoded unitaries, applying these unitaries to the state $\ket{S}$ doesn't change it, the sequence of states in $\ket{S}$ contains each term $\ket{i}\ket{\psi_i}$ with a significant amplitude, the initial and final states $\ket{\psi_1}$ and $\ket{\psi_{m+1}}$ are what we asked for, and that the energy of each state $\ket{\psi_i}$ is low enough.

We show the detailed soundness proof in Section~\ref{sec:soundness}, and
continue with completeness in Section~\ref{sec:completeness}.
Our proof of soundness starts similarly to the one in \cite{BlierTapp}.
In contrast to \cite{BlierTapp}, we require much stronger guarantees on the uniformity of the sequence $\ket{U}$.
We are also asking for an encoding of 1- and 2- qubit gates instead of 3 colors for the graph coloring problem, so the dimension of the gate register has to be $G = \textrm{poly}(n)$.
Next, we have a batch of tests: \textsc{swap~S} and \textsc{sequence}, involving the sequence-encoding state $\ket{S}$. These are new and specific for the shortened quantum proof of GSCON. Finally, the \textsc{start}, \textsc{end}, and \textsc{low} tests check the boundary conditions and the low energy condition for the purported traversal of the low energy space of our ff-GSCON Hamiltonian.

\subsection{Soundness analysis} 
\label{sec:soundness}

Thanks to the promise of the ff-GSCON problem, in the {\em NO} case, the verifier receives a description of a GSCON Hamiltonian $H$, for which there does not exist a sequence of 1- and 2- qubit unitaries $\left\{U_i\right\}_{i=1}^{m}$ with $m=\textrm{poly}(n)$, that would transform the low-energy state $\ket{\psi}$ to a state close to $\ket{\phi}$, while staying in the low-energy subspace. 

Let us see what happens in the case of dishonest provers.
Our goal is to find an upper bound on the probability that the verifier accepts a proof from two malicious, but still unentangled provers. 
We will prove a sequence of Lemmas that together imply that when the provers try to cheat, there is a high enough chance that one of the tests from Section~\ref{sec:tests},
chosen at random would detect this.

\begin{figure}
\begin{center}
	\begin{tabular}{ | l | l | c |}
		\hline
		test & rejection threshold $r_i$ & required onwards from \\ 
		\hline
		1. \textsc{swap~U}     & $r_1 =\frac{\delta^2}{8}
			= \frac{1}{32G^4m^8t^6}$ & Lemma \ref{lem:swapU}
			 \\\hline
		2. \textsc{unique}    & $r_2 = \frac{cx^2}{4}
					= \frac{1}{4Gm^6t^4}$ & Lemma \ref{lem:Udefinition}
				 \\\hline
		3. \textsc{uniform}    & $r_3 = \frac{1}{5Gm^4t^2}$ & Lemma \ref{lem:bigA}
				\\\hline
		4. \textsc{swap~S}     & $r_4 = \frac{z}{4} = \frac{\mu^2}{4m^3}$ & Lemma \ref{lem:Sswap}
				\\\hline
		5. \textsc{sequence}   & $r_5 = \frac{1}{8mG} \frac{z}{4}
					= \frac{\mu^2}{32Gm^4}$ & Lemma \ref{lem:TS}
				\\\hline
		6. \textsc{start} 		 & $r_6	= \left(\frac{1}{2m}-6\mu\right)\frac{h^2}{4}$
		& Lemma \ref{lem:distfailstart} \\\hline
		7. \textsc{end} 			 & $r_7	= \left(\frac{1}{2m}-6\mu\right)\left(
			\frac{(\eta_3+h)^2}{2} - \frac{(\eta_3+h)^4}{8}
		\right)$& Lemma \ref{lem:distfailend}\\\hline
		8. \textsc{low}        & $r_8 = \frac{\eta_2}{8Rm}$ & 
					Lemma \ref{lem:lowe} \\\hline
	\end{tabular}
\end{center}	
\label{tab:prob}
\caption{
The Lemmas in the soudness Section~\ref{sec:soundness} 
assume that if we ran test $i$, it would pass with probability at least $1-r_i$.
Here we list the rejection threshold $r_i$ for each of the 8 tests.
We choose the probability to run Test $i$ 
as $p_i=r_i^{-1}/\sum_{j}r_j^{-1}$, so that $p_i r_i = 1-s'$,
where $s'$ is the final soundness parameter.
The thresholds $r_i$ are expressed using the parameters of the GSCON instance ($m,R,\eta_2,\dots,\eta_4$), as well as 
parameters set in \eqref{deltachoice}, \eqref{tchoice}, \eqref{hset}, and \eqref{muset}:
$\delta = \frac{cx}{2G}$,
$c = \frac{1}{G m^2 t^2}$,
$x = \frac{1}{m^2 t}$,
$z=\frac{\mu^2}{m^3}$, 
$t=\frac{848Gm^2}{\mu^2}$,
$h = \min \left\{
		\frac{\eta_4-\eta_3}{4}, \frac{1}{6}\sqrt{\frac{\eta_2}{R}}
	\right\}$,
	and
	$	\mu = \frac{h^2}{144m(\eta_3+h)}$.
} 
\end{figure}

\subsubsection{Verifying consistency and fullness of the sequence $\ket{U}$}

Our first Test (\textsc{swap~U}) is a \textsc{swap} test on the states $\ket{U}$ and $\ket{U'}$.
Because we know that these states come from unentangled provers, they can be written as 
\begin{align}
	\ket{U} &= \sum_{i=1}^{2m}\alpha_{i}\ket{i}\sum_{j}\beta_{i,j}\ket{j}, \label{Ustate}\\
	\ket{U'} &= \sum_{i=1}^{2m}\alpha'_{i}\ket{i}\sum_{j}\beta'_{i,j}\ket{j},
\end{align}
where $\sum_{i}|\alpha_{i}|^2 = 1$ and $\forall i$, $\sum_{j}|\beta_{i,j}|^2 = 1$, and the same holds for $\alpha_i'$ and $\beta'_{i,j}$.

We will start with showing that passing Test 1 (\textsc{swap~U}) with high enough probability implies the distribution of outcomes when measuring the states $\ket{U}$ and $\ket{U'}$ in the computational basis must be very similar.

\begin{lemma}[Consistency of unitaries]
\label{lem:swapU} Let $\ket{U}$ and $\ket{U'}$ be as defined earlier. If there exists a $k$ and an $l$ such that
\footnote{Note that there are squares in the expression, while \cite{BlierTapp}, Lemma 3.3, has a typo, missing the squares.} 
$\big|  |\alpha_k \beta_{k,l}|^2
-|\alpha'_k \beta'_{k,l}|^2 \big| \geq \delta$, then Test 1 (\textsc{swap~U}) will fail with probability at least $r_1 = \frac{\delta^2}{8}$. 
\end{lemma}

\begin{proof} 
This is Lemma~3.3 from \cite{BlierTapp}, and we repeat the proof.

Let $P_{i,j}=|\alpha_i \beta_{i,j}|^2$ and $Q_{i,j}=|\alpha'_i \beta'_{i,j}|^2$ be the probability distributions when $\ket{U}$ and
$\ket{U'}$ are measured in the computational basis. For any von Neumann measurement, the distances defined below are such
that $D(\ket{U}, \ket{U'}) \ge D(P,Q)$, where $P$ and $Q$ are the classical distributions of the measurement outcomes. Then,
\begin{eqnarray*}
\sqrt{1-|\braket{U}{U'}|^{2}}  
&{\stackrel{\text{\tiny def}}{=}}&  D(\ket{U}, \ket{U'}) 
\geq    D(P,Q)   
\,\, {\stackrel{\text{\tiny def}}{=}} \,\, 
			\half \sum_{ij} \left||\alpha_i \beta_{i,j}|^{2} - |\alpha'_i \beta'_{i,j}|^{2}\right|\\
&\geq &   \half\left||\alpha_k \beta_{k,l}|^{2} - |\alpha'_k \beta'_{k,l}|^{2}\right|
\geq    \frac{\delta}{2},
\end{eqnarray*}
assuming there exist $k,l$ with $\big|  |\alpha_k \beta_{k,l}|^2
-|\alpha'_k \beta'_{k,l}|^2 \big| \geq \delta$.
In that case, $|\braket{U}{U'}|^{2} \leq 1-\frac{\delta^2}{4}$
and Test 1 (\textsc{swap~U})  
will fail with probability
at least $\frac{\delta^2}{8}$.
\end{proof}

Therefore, if Test 1 would pass with probability at least $1-\frac{\delta^2}{8}$,
thanks to Lemma~\ref{lem:swapU}, we get a guarantee on the closeness of $\ket{U}$ and $\ket{U'}$:
\begin{align}
	\big|  |\alpha_k \beta_{k,l}|^2
-|\alpha'_k \beta'_{k,l}|^2 \big| < \delta.
	\label{u1guar}
\end{align}
Let us call $r_1 = \frac{\delta^2}{8}$
the {\em rejection threshold} for Test 1.
We will choose the probability $p_1$ to run Test 1 so that it is tied to the final soundness parameter as $s' = 1- p_1 r_1$.
We set the parameter $\delta$ below in \eqref{deltachoice},
and list $r_1$ in Figure~\ref{tab:prob}.

Let us look at the second test, armed with the guarantee \eqref{u1guar}.
We will prove that passing Test 2 (\textsc{unique}) with high probability means nodes with a high enough probability of being observed encode a well-defined {\em unitary}.
In particular, there is one $\beta_{i,j_i}$ that dominates, and the other $\beta_{i,\dots}$'s are small.

\begin{lemma}[Well defined unitaries]
\label{lem:Udefinition} Assume that the quantum proof would fail Test 1 (\textsc{swap~U}) 
with probability below $r_1=\frac{\delta^2}{8}$,
and fail Test 2 (\textsc{unique}) with probability below $r_2=\frac{cx^2}{4}$ (see also Figure~\ref{tab:prob}).
Then $\forall i: \,|\alpha_{i}|^2 \ge x$, there exists a $j \text{ such that } |\beta_{i,j}|^{2} \ge 1-c$, with 
\begin{align}
	c = \frac{1}{G m^2 t^2}, \quad
	x = \frac{1}{m^2 t}, \quad
	\delta = \frac{cx}{2G} = \frac{1}{2t^3 G^2m^4}, 
	\label{deltachoice}
\end{align}
where $G$ is the number of possible gates
and $t$ is a parameter to be chosen later in \eqref{tchoice}.
\end{lemma}

\begin{proof} 
This is a more general version of Lemma~3.4 from \cite{BlierTapp}, with stronger conditions and implications.

First note, that with the particular string $u_i$ we receive,
we can easily test if it encodes some unitary $U_i$ from the expected gate set.
We reject on failure. 

Now suppose for the sake of contradiction that there exists an $i$ with $\left|\alpha_i\right|^2\geq x$, for which the largest of the $\beta_{i,j}$'s  (without loss of generality, let it be $\beta_{i,0}$) obeys  $|\beta_{i,0}|^2 < 1-c$. Let us then calculate the probability of failing the \textsc{unique} test. It is surely bigger than 
\begin{align}
	p_{\textrm{\textsc{unique}}}^{\textrm{fail}} 
		\geq \left|\alpha_i \beta_{i,0}\right|^2 
			\sum_{j>0} \left|\alpha'_i \beta'_{i,j}\right|^2.
\end{align}
Because of Lemma~\ref{lem:swapU}, we know that $\big|\alpha'_i \beta'_{i,j}\big|^2 \geq \big|\alpha'_i \beta'_{i,j}\big|^2-\delta$. Therefore, we have
\begin{align}
	p_{\textrm{\textsc{unique}}}^{\textrm{fail}}  
		&\geq x |\beta_{i,0}|^2 
		\bigg(
			\sum_{j>0} \left|\alpha'_i \beta'_{i,j}\right|^2 - \left(G-1\right)\delta
		\bigg)
		\\
		&= x^2 |\beta_{i,0}|^2  \left(1-|\beta_{i,0}|^2 \right)
		- x\left(G-1\right)\delta |\beta_{i,0}|^2 \\
	&\geq x^2 (1-c)c - x\left(G-1\right)\delta (1-c) \\
	&\geq (1-c)x \left(cx-\delta G\right).
\end{align}
We set the parameters $c, x, \delta$ according to \eqref{deltachoice}, with a large
$t$ chosen later in \eqref{tchoice}. This gives us a bound
\begin{align}
		p_{\textrm{\textsc{unique}}}^{\textrm{fail}}  
			&\geq 
			\left(1-\frac{1}{Gm^2t^2}\right)\frac{cx^2}{2}
			\geq \frac{cx^2}{4},
			\label{pinco}
\end{align}
proving the Lemma.
\end{proof}

The rejection threshold $r_2$ and the probability $p_2$ to run Test 2 
are listed in Figure~\ref{tab:prob}, and chosen so that 
if the combined probability of passing the \textsc{unique} test
is at least $s'$, we get a guarantee on how well the $U$'s are defined in $\ket{U}$ from Lemma~\ref{lem:Udefinition}:
\begin{align}
	\forall i, \textrm{for which } |\alpha_{i}|^2 \ge x \textrm{, } \exists! j \text{ such that } |\beta_{i,j}|^{2} \ge 1-c,
	\label{u2guar}
\end{align}
with $c=(Gm^2t^2)^{-1}$ for $t$ from \eqref{tchoice}.

Armed with \eqref{u2guar}, let us look at the third test.
The next three Lemmas quantify what passing the tests up to and including Test 3 (\textsc{uniform}) with high probability implies: the state $\ket{U}$ contains a nearly uniform superposition of states of the form $\ket{i}\ket{U_i}$.
We start by showing that the probability to find a uniform superposition in the gate (second) register of $\ket{U}$, when performing the first measurement of Test 3, is very well defined.

\begin{lemma}[Projection onto the uniform superposition of gates]
\label{lem:fourier} Assume the quantum proof would fail Test 1 (\textsc{swap~U}) and Test 2 (\textsc{unique}) 
with probabilities below $r_1$ and $r_2$ from Figure~\ref{tab:prob}.
Then the probability of measuring $\ket{\bar{g}}=F_{G} \ket{0}$ in the Fourier basis on the gate register is at least $\frac{1}{G}\left(1-\frac{6}{mt}\right)$ for large enough $m$.
Moreover, for each $i$ with $\left|\alpha_i\right|^2 \geq x$, 
the individual probability of this projection satisfies
$|p_i^{F_G}-\frac{1}{G}|\leq\frac{4}{Gmt}$.
\end{lemma}

\begin{proof} 
This Lemma is based on Lemma~3.5 from \cite{BlierTapp}, and has much stronger conditions and implications.

Thanks to the assumption on the rejection probability for the previous tests, we can use Lemmas~\ref{lem:swapU} and~\ref{lem:Udefinition}.

Assume that the first (label) register of the state $\ket{U}$ is measured. If the outcome is $i$, then the probability of obtaining $\ket{\bar{g}}$ in the Fourier
basis on the gate register is given by $p_i^{F_G} = \frac{1}{G}\big|\sum_j \beta_{i,j}\big|^2$.
For all $i$ with $\left|\alpha_i\right|^2 \geq x$, Lemma \ref{lem:Udefinition} applies, in which case we can assume w.l.o.g that $|\beta_{i,0}|^2>1-c$ and
$\sum_{j\neq 0} |\beta_{i,j}|^2 \leq c$. Using the Cauchy-Schwarz inequality, we obtain
\begin{align} 
	p_i^{F_G} = \frac{1}{G} \bigg|\sum_j \beta_{i,j}\bigg|^2
  &\geq  \frac{1}{G}\Bigg| 
 \left|\beta_{i,0}\right| - \bigg|\sum_{j\neq 0} \beta_{i,j}\bigg|
\Bigg|^2
  \geq  \frac{1}{G}\Bigg| 
 \left|\beta_{i,0}\right| - \sqrt{G \sum_{j\neq 0} \left|\beta_{i,j}\right|^2}
\Bigg|^2
 \nonumber\\
& \geq  \frac{1}{G}\left| 
 \sqrt{1-c} - \sqrt{G c}
\right|^2
\geq \frac{1}{G}\left|1-\frac{1}{m^2Gt^2}-\frac{1}{mt}\right|^2
 \geq 
\frac{1}{G}\left(1-\frac{4}{mt}\right), \label{Fbelow}
 \end{align}
for $c=(Gm^2t^2)^{-1}$.

Note that in $\ket{U}$ \eqref{Ustate}, at least one $\left|\alpha_i\right|^2 \geq x$, or equivalently, at most $2m-1$ can obey $\left|\alpha_i\right|^2<x$ so that Lemma~\ref{lem:Udefinition} doesn't apply to them. Therefore, when projecting the gate register of the whole state $\ket{U}$ onto the uniform superposition, the probability of obtaining 0 is at least
\begin{align}
    \left(1-(2m-1) x\right)\frac{1}{G}\left(1-\frac{4}{m}\right) 
			\geq \left(1-\frac{2}{mt}\right)\left(1-\frac{4}{mt}\right) \frac{1}{G}
			\geq \left(1-\frac{6}{mt}\right)\frac{1}{G}.
\end{align}

In addition to \eqref{Fbelow}, we can also find an upper bound on the individual probabilities $p_{i}^{F_G}$.
For $i$ with $|a_i|^2\geq x$, Lemma~3 applies, and one of the $\beta_{i,j}$'s is necessarily large. The probability for a successful projection onto a uniform superposition 
is then bounded from above by a situation where the $\beta$'s are as balanced as possible:
\begin{align}
p_i^{F_G} = \frac{1}{G} \bigg| \sum_j \beta_{i,j}\bigg|^2
= \frac{1}{G} \bigg| \beta_{i,0} + \sum_{j>0} \beta_{i,j}\bigg|^2
  &\leq  
	\frac{1}{G}\left(
	 	(1-c) + \left(G-1\right) \sqrt{\frac{1-(1-c)^2}{G-1}} 
		\right)^2
		\nonumber\\
  &\leq \frac{1}{G}\left(1+2\sqrt{Gc}\right)^2
  \leq \frac{1}{G}\left(1+\frac{4}{mt}\right).
\end{align}
This concludes the proof of Lemma~\ref{lem:fourier}.
\end{proof}

Thus, if Test 1 and Test 2 are likely to succeed, Part 1 of Test 3 (the Fourier projection on the gate register) will succeed with probability at least $\frac{1}{G}\left(1-\frac{6}{mt}\right)$, allowing us to continue to the second step of Test 3. 
It involves a measurement of the label register that detects slightly non-uniform states.

Let us look on the state $\ket{U}$ after the projection on the uniform superposition of gates. We can write this projected and normalized state as $\sum_i \gamma_i \ket{i} \ket{\bar{g}}$.
The following Lemma tells us that to successfully pass the Fourier-basis projection onto the uniform superposition of states in the label register, the $\gamma_i$'s all have to be very close to $\frac{1}{\sqrt{2m}}$ in magnitude.
\begin{lemma}[A successful Fourier projection implies uniformity]
\label{lem:Fourproject} Given a state $\ket{X} = \sum_{i} \gamma_{i}\ket{i}$ such that there exists an $l$ with $\left||\gamma_{l}|^2 - \frac{1}{2m}\right| > \frac{f}{m}$,
the probability of not getting $\ket{\overline 0}=F_{2m} \ket{0}$ when we measure $\ket{X}$ in the Fourier basis is greater than $\frac{f^2}{4m^{2}}$,
for $f>0$.
\end{lemma}

\begin{proof} 
This is a stronger version of Lemma~3.6 from \cite{BlierTapp}.

The probability of not getting $\ket{\bar{0}}$ when measuring $\ket{X}$ depends on the overlap of these states. Let us call $P$ and $Q$ the probability distributions for a computational basis measurement of $\ket{X}$ and $F_{2m} \ket{0}$, respectively. 
The probability of not getting $\ket{\bar{0}}$ then obeys
\begin{eqnarray}
1-|\braket{X}{\overline 0}|^2 
&= \left(D(\ket{X}, \ket{\overline 0}) \right)^2
\ge \left(D(P,Q)\right)^2
= 
\left(\half \sum_{i} \left| \prob{\gamma_{i}} - \frac{1}{2m} \right| \right)^2
\ge \left(\half \left| |\gamma_{l}|^2 - \frac{1}{2m} \right| \right)^2
 > \frac{f^2}{4m^2}.
\end{eqnarray}
\end{proof}

Lemma~\ref{lem:Fourproject} allows us to prove a statement about the original coefficients $\alpha_i$ in $\ket{U}$: passing Tests 1-3 with high probability implies a valid encoding of all the required unitaries $U_i$ for $i=1,\dots 2m$, with nearly uniform prefactors, as stated in the next Lemma.

\begin{lemma}[A full sequence of unitaries]
\label{lem:bigA} 
Assume that Test 1 (\textsc{swap~U}) and Test2 (\textsc{unique})
fail with probability below $r_1$ and $r_2$ from Figure~\ref{tab:prob}.
Assume that Test 3 (\textsc{uniform}) fails with probability below $r_3 = \frac{1}{5Gm^4 t^2}$.
Then the coefficients $\alpha_i$ in the state $\ket{U}$ obey
$\left| \left|\alpha_i\right|^2 -\frac{1}{2m}\right| \leq \frac{13}{2m^2t}$,
for all $i$.
\end{lemma}

Note that the parameter $t$ is still free. We will set it to be a large number
later \eqref{tchoice}.

\begin{proof} 
Thanks to Lemmas~\ref{lem:swapU}-\ref{lem:Fourproject}, we are now able to show a bound on the coefficients $\alpha_i$ that is much tighter than Lemma 3.7 in \cite{BlierTapp}.

Thanks to the assumption on the rejection probabilities for the previous tests, we can use
the previous Lemmas. We also add the assumption that Test 3 rejects the proof with probability below $\frac{1}{5Gm^4 t^2}$. This rejection can happen only if the first Fourier projection on the gate register passes (this has probability at least $\frac{1}{G}\left(1-\frac{6}{mt}\right)$ according to Lemma~\ref{lem:fourier}), and then the second Fourier projection on the label register fails.
When we choose $f=\frac{1}{mt}$ in Lemma~\ref{lem:Fourproject}
we see that the second Fourier basis projection either rejects with probability at least $\frac{1}{4m^4t^2}$,
or we get a guarantee that no $|\gamma_l|^2$ is farther from $\frac{1}{2m}$ than $\frac{f}{m}$.
The overall probability of detecting a cheater is now thus either at least
$\frac{1}{G}\left(1-\frac{6}{mt}\right)\frac{1}{4m^4t^2}
\geq \frac{1}{5Gm^4 t^2}$, a contradiction on the assumption of the Lemma,
or we get the guarantee on $|\gamma_l|^2$.

Let us then work with this guarantee
and analyze what happens after the first successful projection onto the uniform superposition $\ket{\bar{g}}$ in the gate register of $\ket{U}$, i.e. the first step of Test 3. 
For the significant $\alpha_i$'s ($\left|\alpha_i\right|^2 \geq x$), Lemma~\ref{lem:Udefinition} tells us that they encode a pretty well 
defined unitary, and Lemma~\ref{lem:fourier} tells us that the probability of getting a successful projection onto $\ket{\bar{g}}$ for each of these $i$'s is at most $\frac{4}{Gtm}$ far from $\frac{1}{G}$.
This projection thus brings down the norm of this part of the state, but not to something smaller than 
\begin{align}
	\|\textrm{the large-$\alpha_i$ part after the projection}\|^2
	 \geq \frac{1}{G} \left(1-\frac{4}{tm}\right)
\sum_{\left|\alpha_i\right|^2\geq x} \left|\alpha_i\right|^2.
\end{align}
Next, we know there can't be too much of the norm of the state $\ket{U}$ hiding in parts of the superposition with small $\left|\alpha_i\right|^2 \leq x$. The state $\ket{U}$ is normalized, and there are at most $2m-1$ such $i$'s, so the norm of that small-$\alpha_i$ part of the state is 
\begin{align}
	\sum_{\left|\alpha_i\right|^2 <x}  \left|\alpha_i\right|^2 \leq (2m-1)x \leq \frac{2}{tm}, \label{smallalpha}
\end{align}
for our choice of $x=\frac{1}{m^2t}$ in \eqref{deltachoice}.

Even if the projection on the uniform superposition in the gate register kills this small-$\alpha_i$ part, the overall norm squared $N^2$ of the whole state after the projection is at least
\begin{align}
	N^2 
	&\geq
	\frac{1}{G} \left(1-\frac{4}{tm}\right) \sum_{\left|\alpha_i\right|^2\geq x} 	 \left|\alpha_i\right|^2
		\geq 
	\frac{1}{G} \left(1-\frac{4}{tm}\right)\left(1-\frac{2}{tm}\right)
	\geq \frac{1}{G} \left(1-\frac{6}{tm}\right), \label{normbig}
\end{align}
using \eqref{smallalpha} and our choice \eqref{deltachoice}.

Let us find a stronger lower bound for $\left|\alpha_i\right|^2 \geq x$.
We obtain the $\gamma_i$'s by normalizing the state after the projection.
Using \eqref{normbig} and recalling the large-$|\alpha_i|^2$ terms are multiplied by at most $\frac{1}{G}\left(1+\frac{4}{tm}\right)$ when projected, we obtain
\begin{align}
	|\gamma_i|^2 
	\leq \frac{\frac{1}{G}\left(1+\frac{4}{mt}\right)}{\frac{1}{G}\left(1-\frac{6}{mt}\right)} \left|\alpha_i\right|^2
	\leq \left(1+\frac{11}{mt}\right) \left|\alpha_i\right|^2,
\end{align}
for large enough $m, t$.
Because we know from Lemma~\ref{lem:Fourproject} that all $|\gamma_i|^2$ must be close to $\frac{1}{2m}$, those $\left|\alpha_i\right|^2 \geq x$ must obey
\begin{align}
	\frac{1}{2m}\left(1-2f\right) &\leq |\gamma_i|^2 
	\leq \left(1+\frac{11}{mt}\right) \left|\alpha_i\right|^2. 
\end{align}
Choosing $f=\frac{1}{mt}$ in Lemma~\ref{lem:Fourproject}, we have
\begin{align}
	\left|\alpha_i\right|^2 & \geq \frac{1}{2m} \left(\frac{1-\frac{2}{mt}}{1+\frac{11}{mt}}\right)
	\geq \frac{1}{2m} \left(1-\frac{13}{mt}\right). \label{bigA}
\end{align}

What about the small $\left|\alpha_i\right|^2 < x$? Even if they do not decrease on projection, they get multiplied by at most $\frac{1}{G}\left(1+\frac{4}{tm}\right) \leq \frac{2}{G}$, which implies $|\gamma_i|^2 \leq \frac{2x}{G}$.
However, because $x = \frac{1}{tm^2}$, such $|\gamma_i|^2$ would be much smaller than $\frac{1}{2m}$, and thus easily detectable by Lemma~\ref{lem:Fourproject}.
Therefore, small $\left|\alpha_i\right|^2<x$ can not exist in the superposition $\ket{U}$ without being detected by our tests with a reasonable probability. 

Therefore, {\em all} $\left|\alpha_i\right|^2$ are bounded from below by \eqref{bigA}.
Moreover, we can also find a limit on how big they can be. 
To show this, we start with an upper bound on the norm of the whole state after the projection.
\begin{align}
	N^2 = \|\textrm{the whole state after the projection}\|^2
	 \leq 
	\frac{1}{G} \left(1+\frac{4}{mt}\right)
\sum_{\left|\alpha_i\right|^2\geq x} \left|\alpha_i\right|^2
	= \frac{1}{G} \left(1+\frac{4}{mt}\right),
\end{align}
as there are no small-$\alpha_i$ coefficients.
This implies for the $\gamma_i$'s that
\begin{align}
	|\gamma_i|^2 \geq 
		 \frac{\frac{1}{G}\left(1-\frac{4}{mt}\right)}{\frac{1}{G}\left(1+\frac{4}{mt}\right)}
		\left|\alpha_i\right|^2 \geq
		\left(1-\frac{8}{mt}\right) \left|\alpha_i\right|^2.
\end{align}
Recalling the guarantee $\left||\gamma_i|^2 -\frac{1}{2m}\right| \leq \frac{f}{m}$ from Lemma~\ref{lem:Fourproject} with $f=\frac{1}{mt}$, we also obtain
\begin{align}
	\frac{1}{2m}\left(1+2f\right) &\geq |\gamma_i|^2 
	\geq \left(1-\frac{8}{mt}\right) \left|\alpha_i\right|^2, 
\end{align}
which translates to an upper bound on $|\alpha_i|^2$:
\begin{align}
	\left|\alpha_i\right|^2 
	&\leq \frac{1}{2m} \left(\frac{1+\frac{2}{mt}}{1-\frac{8}{mt}}\right)
	 \leq \frac{1}{2m} \left(1+\frac{13}{mt}\right),
		\label{bigA2}
\end{align}
for large enough $m, t$.

Putting together \eqref{bigA} and \eqref{bigA2} finishes the proof.
Note that all $|\alpha_i|^2$ are thus large enough for Lemma 4, so {\em all} of the $2m$ encoded unitaries must be ``well defined''.
\end{proof}

Therefore, if we chose to run Tests 1-3 on $\ket{U}$ and $\ket{U'}$, and each would be likely to pass, we have a guarantee that the state $\ket{U}$ as well as the state $\ket{U'}$ must have form very close to what we demand, i.e.
\begin{align}
	\ket{U} 
	&= \frac{1}{\sqrt{2m}} \sum_{i=1}^{2m} \ket{i}\ket{u_i}  
	+ \frac{1}{\sqrt{2m}} \sum_{i=1}^{2m} \theta_i \ket{i}
	\ket{\theta_i},
	\label{Ugood}
\end{align}
where $u_i$ are computational basis states that encode the gates $U_i$, 
the second (error) term is orthogonal to the first one,
and 
\begin{align}
	|\theta_i|^2 \leq 2\left(\frac{26}{tm}+\frac{1}{Gm^2t^2}\right) \leq \frac{53}{tm},
	\label{Ugood2}
\end{align}
where the first term comes from $|\alpha_i|^2$ possibly deviating from $\frac{1}{2m}$, and the second term from the possible imprecision in the definition of the unitaries (the error $c$ in Lemma~\ref{lem:Udefinition}).
This encoding of the unitaries is solid enough to help us verify the state sequence is also proper, and thus prove the soundness of our verifier.

\subsubsection{Verifying consistency of the states $\ket{S}$}

We will now show how to apply the $U_i$'s to the state $\ket{S}$, in order to test if it is a proper cyclical sequence connected by 1 and 2 qubit gates.
It requires a guarantee on the consistency of the $\ket{S}$ states,
and a procedure for the probabilistic application of the $U_i$'s.

Let us quantify what the \textsc{swap~S} test (Test 4) implies for 
the similarity of two witness states $\ket{S}$ and $\ket{S'}$.
\begin{lemma}[State consistency]
\label{lem:Sswap}
Let $\ket{S}$ and $\ket{S'}$ be two-register, normalized quantum states
\begin{align}
\ket{S} = \sum_{i=1}^{m}a_{i}\ket{i}\ket{\psi_i},\quad\ket{S'} = \sum_{i=1}^{m}a'_{i}\ket{i}\ket{\psi'_i},
\end{align}
and label 
$\ket{\Delta_S}=\ket{S}-\ket{S'}=\sum_{i}\ket{i}\ket{\delta_i}$,
with $\ket{\delta_i}=a_{i}\ket{\psi_i}-a'_{i}\ket{\psi'_i}$.
If there exists a $k$ such that $\braket{\delta_k}{\delta_k}\geq z$, the \textsc{swap~S} test (Test 4) on the states $\ket{S}$ and $\ket{S'}$ will fail with probability at least $r_4=\frac{z}{4}$.
\end{lemma}
\begin{proof}
Without loss of generality, we can assume the phase of $\ket{S'}$ is such that $\braket{S}{S'}\in\RR$, as $\ket{S}$ and $\ket{S'}$ come from two unentangled provers. This lets us write 
\begin{align}
\braket{S}{S'}&=1-\frac{1}{2}\braket{\Delta_S}{\Delta_S}
=1-\frac{1}{2}\sum_{i}\braket{\delta_i}{\delta_i}\leq 1-\frac{1}{2}\braket{\delta_k}{\delta_k}\leq 1-\frac{z}{2},
\end{align}
instead of having to deal with absolute values or real/imaginary parts.
This translates to 
$|\braket{S}{S'}|^2\leq(1-\frac{z}{2})^2 \leq 1-\frac{z}{2}$ and 
the probability to fail Test 4 (the \textsc{swap~S} test) 
$\frac{1}{2}\left(1-|\braket{S}{S'}|^2\right) \geq \frac{z}{4}$.
\end{proof}

We will later choose $z$ to be a small number \eqref{tchoice}.
Similarly to the previous tests, we will demand that combined with the probability $p_4$ to run Test 4, the probability to detect a cheating Merlin is at least $1-s' = p_4 r_4$, or we get the guarantee that for all $i$, $\braket{\delta_i}{\delta_i} <z$, with $z$ chosen in Lemma~\ref{lem:TS} \eqref{tchoice}.

Could we continue with something similar to the \textsc{uniform} test?
The size of the state space for the $\ket{\psi_i}$'s is too large, 
and we don't know enough about the states to ensure a reasonable chance of success for the projection onto a uniform superposition. 
Instead, we will use the state $\ket{U}$ to probabilistically apply the unitary $W$ \eqref{Wunitary} to the state $\ket{S}$ and compare it with $\ket{S'}$.
This \textsc{sequence} test (Test 5) checks whether $\ket{S}$ and $\ket{S'}$ contain
a balanced enough superposition corresponding to a cyclical sequence of states connected by the 1- and 2- local gates $U_i$.

Let us look at the probabilistic procedure
described in detail in the definition of Test 5.
We apply the gates from $\ket{U}$ to $\ket{S}$, 
project onto an uniform superposition in the gate register, drop it, 
project onto identical label registers, uncompute and drop one of them, and shift the remaining label up by one. This should prepare
\begin{align}
	\ket{T} = \frac{1}{\sqrt{2m}}\sum_{i} \ket{i+1}U_i\ket{\psi_i},
	\label{Tshift}
\end{align}
which we want to \textsc{swap} test with the state $\ket{S'}$. 
However, we need to deal with dishonest Merlins. We know that if the previous tests pass with high enough probability, the unitaries are pretty uniformly encoded, and pretty well defined. Let us now prove a series of Lemmas:  if further tests are very likely to pass, the projections in the cycle consistency test will succeed with reasonable probability, the state $\ket{T'}$ we get in reality is close to the expected state $\ket{T}$ \eqref{Tshift}, and the final \textsc{swap} test in Test 5 (\textsc{sequence}) is strong enough to guarantee proper form of the cyclical sequence, connected by 1- and 2- qubit gates.

\begin{lemma}[Probabilistic gate application]
\label{lem:probgate}
	Let us assume all previous tests (\textsc{swap~U}, \textsc{unique}, \textsc{uniform}, \textsc{swap~S}) would fail
	with respective probabilities below $r_1,\dots,r_4$, as listed in Figure~\ref{tab:prob}. Consider the above procedure that starts with $\ket{U}\ket{S}$, applies the gates from $\ket{U}$ to $\ket{S}$, projects onto the uniform superposition in the gate register, and projects onto identical labels. The joint probability of success for the projections is at least $\frac{1}{8mG}$.
	Moreover, after dropping the extra registers and shifting the label register, the resulting state $\ket{T'}$ is close to the state 
	$\ket{T_a} = \sum_{i} a_i \ket{i+1}U_i\ket{\psi_i}$,
	with the coefficients $\xi_j$ in  
	$\ket{T'}-\ket{T_a} = \sum_{j} \xi_j a_j \ket{j+1} \ket{\xi_j}$
	obeying $|\xi_j|^2\leq \frac{848 G}{tm}$.
\end{lemma}
Note that the parameter $t$ is still free, we set it later in \eqref{tchoice}.
\begin{proof}
Assuming the previously discussed tests would pass with high enough probability
allows us to use the previous Lemmas. In particular, the state $\ket{U}$ must obey \eqref{Ugood} and \eqref{Ugood2}. 

Let us follow the procedure for Test 5 from Section~\ref{sec:tests}. We apply the gates encoded in $\ket{U}$ to the second register state of the state $\ket{S}$ and obtain
\begin{align}
	\frac{1}{\sqrt{2m}} \sum_{i=1}^{2m} \ket{i}\ket{u_i} \sum_{j} a_j \ket{j} U_i \ket{\psi_j} 
	+ \frac{1}{\sqrt{2m}} \sum_{i=1}^{2m} \theta_i 
	\ket{i}\ket{\theta_i}	\sum_{j} a_j \ket{j} \Theta_i \ket{\psi_j}, 
\end{align}
where $U_i$ are the gates described by the computational basis states $\ket{u_i}$, and the prefactor in the error term obeys $|\theta_i|^2 \leq \frac{53}{tm}$.

Note that we here assume perfect application of the gates $U_i$.
This is possible, if they come from a specific universal gate set under our control.
The GSCON problem remains $\QCMA$ complete also under this assumption (as $\QCMA$ verification circuits can come from a specific universal gate set). 
On the other hand, what if we only have access to a smaller universal gate set? We would then have to decompose the $U_i$'s into this set (on the fly), and would get a small error along the way. However, this error can be controlled to whatever inverse polynomial in $n$ we require, so we do not need to consider it here.

Let us now apply the projection of the gate-register onto the uniform-superposition $\ket{\bar{g}}$, and renormalize the state. 
We know that for basis states $\ket{u_i}$, we have
$\braket{\bar{g}}{u_i}=1/\sqrt{G}$, while the second (error) term will increase in importance the most if we assume $\braket{\bar{g}}{z_i}=1$ 
and $\bra{\psi_{j}} \Theta_i^\dagger U_i\ket{\psi_j}=-1$. The norm squared of the state after this projection is at least
\begin{align}
	N^2_{\bar{g}} \geq  2m \frac{1}{2m} \left(\frac{1}{\sqrt{G}}-\theta_{\textrm{max}}\right)^2
	 = \frac{1}{G} \left(1-\theta_{\textrm{max}}\sqrt{G}\right)^2,
\end{align}
which then translates into a normalized state
\begin{align}
	\frac{1}{\sqrt{2m}} \sum_{i=1}^{2m} \ket{i} \sum_{j} a_j \ket{j} U_i \ket{\psi_j} 
	+ \frac{1}{\sqrt{2m}} \sum_{i=1}^{2m} \nu_i \ket{i}
	\ket{\nu_i},
\end{align}
where $|\nu_i|^2 \leq 4 G |\theta_{\textrm{max}}|^2
\leq \frac{212 G}{tm}$, after dropping the gate register, which is in the state $\ket{\bar{g}}$.
We also note that the probability of a successful projection is not smaller than $\frac{1}{2G}$.

Next, we can perform a projection onto identical labels $i=j$.
For a fixed $j$, the probability of this happening is $\frac{1}{2m}$
in the first part of the state,
and $\frac{|\nu_j|^2}{2m}$ in the second part of the state.
Even if the two parts of the state were not orthogonal for a fixed $j$,
we have a guarantee that the probability of a proper projection 
is (for each $j$)
\begin{align}
	\left|p_{i=j} - \frac{1}{2m}\right| \leq \frac{4|\nu_{\textrm{max}}|^2}{2m},
\end{align}
i.e. the norm squared is guaranteed to be within $\frac{4|\nu_{\textrm{max}}|^2}{2m}$ of $\frac{1}{2m}$. After normalization, this translates to a new state
\begin{align}
	\sum_{j=1}^{2m} a_j \ket{j}\ket{j} U_j \ket{\psi_j} 
	+ \sum_{j=1}^{2m} \xi_j a_j \ket{j}\ket{j} \ket{\xi_{j}}, 
\end{align}
with $|\xi_j|^2 \leq 4|\nu_{\textrm{max}}|^2 \leq 
16 G |\theta_\textrm{max}|^2 \leq \frac{848 G}{tm}$.
Let us note that the probability of this successful projection is surely not smaller than $\frac{1}{4m}$. Overall, the probability of passing both projections successfuly is surely no smaller than $\frac{1}{8Gm}$.

Uncomputing and dropping one of the label registers is then simple.
We also shift the remaining label register up by one.
All in all, with probability at least $\frac{1}{8mG}$,
the procedure described above results in the state
\begin{align}
	\sum_{j=1}^{2m} a_j \ket{j+1} U_j \ket{\psi_j} 
	+ \sum_{j=1}^{2m} \xi_j a_j \ket{j+1} \ket{\xi_{j}}, 
\end{align}
with normalized states $\ket{\xi_j}$, and a guarantee $|\xi_j|^2 \leq \frac{848 G}{tm}$, as claimed in the Lemma.
\end{proof}

Therefore, when Tests 1-5
are likely to pass (as described in the conditions of the previous Lemmas), the state $\ket{U}\ket{S}$ after a succesful transformation, projection, label dropping and shift can be written as
\begin{align}
	\ket{T'} = \sum_j a_j \ket{j+1}  \left( U_j \ket{\psi_j} + \xi_j \ket{\xi_j}\right),
	\label{Tguarantee}
\end{align}
with a guarantee $|\xi_j|^2 \leq \frac{848 G}{tm}$ on the error terms.
With this in mind, we can turn to the last step in Test 5: the \textsc{swap} test between $\ket{T'}$ and $\ket{S'}$.
The goal of the next Lemma is to show that if this \textsc{swap} is likely to pass, the states $\ket{S}$ and $\ket{S'}$ must encode a reasonably uniform superposition of states -- the whole sequence of low-energy states connected by the gates $U_i$.

Note that we don't (need to) verify that the sequence of unitaries $U_i$ actually computes and uncomputes the transformation from $\ket{\psi_1}$ to $\ket{\psi_{m+1}}$. We only check if the whole sequence is cyclically invariant under the transformation \eqref{Wunitary}, i.e. that $\ket{\psi_{j+1}} =U_j\ket{\psi_j}$ and that $U_{2m}\dots U_1 =\ii$.

\begin{lemma}
	\label{lem:TS}
	Assume the previous tests (\textsc{swap~U}, \textsc{unique}, \textsc{uniform}, \textsc{swap~S})
	would fail with respective probabilities below $r_1,\dots,r_4$, listed in Figure~\ref{tab:prob},
	with $t=\frac{848Gm^2}{\mu^2}$ and $z=\frac{\mu^2}{m^3}$.
	If the \textsc{sequence} test rejects the proof with probability
	below $r_5= \frac{\mu^2}{32Gm^4}$,
	for a small $\mu$ to be set later \eqref{muset},
	we claim that the original state $\ket{S}$ obeys
	\begin{align}
		|a_j|^2 &\geq \frac{1}{2m}- 6\mu, 
		\qquad \textrm{and} \qquad
		\norm{\ket{\psi_{j+1}} - U_j U_{j-1} \dots U_2 U_1 \ket{\psi_1}}\leq \frac{6j \mu}{m}. \label{bigAA}
	\end{align}
	Thus also,
	\begin{align}
		\norm{\ket{\psi_{m+1}} - U_m \dots U_1 \ket{\psi_1}}\leq 6 \mu.
		\label{finalstatebound}
	\end{align}
\end{lemma}

Note that we choose $t, z$ here, but tie them to another parameter, $\mu$,
which we later in \eqref{muset} choose as a function of the GSCON problem instance parameters (see Definition~\ref{def:GSCON}).
In particular, we set it so that $6\mu \leq h \leq \frac{\eta_4-\eta_3}{4}$.

\begin{proof}
We assume the 
previous tests would fail with probabilities below the $r_i$'s listed in Figure~\ref{tab:prob}, 
so the previous Lemmas apply. We also assume the \textsc{sequence} test rejects the proof with probability below $r_5 = \frac{\mu^2}{32Gm^4}$.
Thanks to Lemma~\ref{lem:probgate}, we know the probabilistic preparation of $\ket{T'}$ from $\ket{S}$ and $\ket{U}$ according to the description in the \textsc{sequence} test succeeds with probability at least $\frac{1}{8mG}$.
Therefore, the subsequent \textsc{swap} test between $\ket{T'}$	nd $\ket{S'}$ must not reject with probability above $\frac{\mu^2}{4m^3}$. 
Let us unravel what it implies for the state $\ket{S}$.

Recall that $\ket{S'}=\sum_j a'_j \ket{j} \ket{\psi'_j}$. 
Lemma~\ref{lem:Sswap} with a parameter $z$ about a \textsc{swap} test between $\ket{S}$ and $\ket{S'}$,
says that the norm squared of
$\ket{\delta_{j+1}} =a_{j+1}\ket{\psi_{j+1}}-a'_{j+1}\ket{\psi'_{j+1}}$ is below $z$.
Similarly, we can apply the procedure from Lemma~\ref{lem:Sswap}
to a \textsc{swap} test between $\ket{T'}$ and $\ket{S'}$.
Recalling the previous result \eqref{Tguarantee}, we can write
\begin{align}
	\ket{T'} -\ket{S'} = \sum_j \ket{j+1} \Big(
	\underbrace{a_j U_j \ket{\psi_j} + a_j \xi_j \ket{\xi_j} -a'_{j+1}\ket{\psi'_{j+1}}}_{\ket{y_j}}
	\Big).
\end{align}
Thus, if the \textsc{swap} test between $\ket{T'}$ and $\ket{S'}$ succeeds with probability at least $1-\frac{z}{4}$, then for any $j$, we have
$\braket{y_j}{y_j} \leq z$.

Let us combine these facts and use the triangle inequality to derive:
\begin{align}
	\norm{a_j U_j \ket{\psi_j}-a_{j+1}\ket{\psi_{j+1}}}
	\leq \norm{\ket{y_j}}
	+ \norm{a_j \xi_j \ket{\xi_j}}
	+ \norm{a'_{j+1}\ket{\psi'_{j+1}}-a_{j+1}\ket{\psi_{j+1}}}
	\leq 2\sqrt{z} + \sqrt{\kappa_j}, \label{sequencediff}
\end{align}
where $\kappa_j = |a_j \xi_j|^2 \leq \frac{848 G}{tm}|a_j|^2$.
Note that the left side is the smallest for real positive $a_j,a'_j$ and $U_j \ket{\psi_j} = \ket{\psi_{j+1}}$, which can be rewritten as $\left||a_j|-|a_{j+1}|\right| \leq \norm{a_j U_j \ket{\psi_j}-a_{j+1}\ket{\psi_{j+1}}}$.
Therefore, when we take into account what we know about $\xi_j$,
we obtain $\left||a_j|-|a_{j+1}|\right| \leq 2\sqrt{z}+\sqrt{\frac{848 G}{tm}}|a_j|$.
Now, at least one of the $|a_j|$'s has to be at least $\frac{1}{\sqrt{2m}}$, as $\sum_{j=1}^{2m} |a_j|^2=1$. 
Let us see how small could some other $|a_j|$ be, as it must be tied to the neighboring ones by what we proved above. In $m$ steps away from the specific large $a_k$, all of the $a_j$ have to obey (w.l.o.g. assuming positive $|a_j|$ and dropping the absolute values)
\begin{align}
	a_{j+1} \geq a_j \left(1-\sqrt{\frac{848 G}{tm}}\right) - 2\sqrt{z}.
\end{align}
Doing this $m$ times and assuming a large $m$, 
labeling $v=1-\sqrt{848G/tm}$, we get
\begin{align}
	a_{j+m} &\geq a_j v^m - 2\sqrt{z} \left(1+v+v^2+\dots+v^{m-1}\right)
	\geq a_j \left(1- \sqrt{\frac{848 Gm}{t}}\right) - 2 m \sqrt{z}.
	\label{ajalmostdone}
\end{align}
We now choose a small enough $z$ and a large enough $t$:
\begin{align}
	z=\frac{\mu^2}{m^3}, \qquad
	t = \frac{848G}{m z} = \frac{848 Gm^2}{\mu^2}.
	\label{tchoice}
\end{align}
parametrized by a new free parameter $\mu$, which we later \eqref{muset} choose according to the parameters $\eta_2, \eta_3, \eta_4$ from the GSCON problem instance.
For small $\mu$, we have $\sqrt{\frac{848 Gm}{t}} = \frac{\mu}{\sqrt{m}}$, and $2m\sqrt{z} = \frac{2\mu}{\sqrt{m}}$.
When we use it in \eqref{ajalmostdone}, together with $|a_j|\geq \frac{1}{\sqrt{2m}}$, we obtain $|a_{j+m}| \geq \frac{1}{\sqrt{2m}}
\left(1- 6m\mu \right)$. This implies what we wanted to prove for all $i$:
\begin{align}
	|a_i|^2 &\geq \frac{1}{2m}\left(1-12m\mu\right) = \frac{1}{2m}-6\mu,
	\label{Abound}
\end{align}
i.e. all the coefficients $a_i$ have to be very close to $\frac{1}{2m}$ (for small $\mu$), and thus significant.

We can now prove that the state $\ket{S}$ is made from a {\em sequence} of states close to $\ket{\psi_{j+1}} = U_j \dots U_1 \ket{\psi_1}$.
Combining \eqref{sequencediff} with \eqref{Abound},
and using the triangle inequality $j$ times, we get
\begin{align}
	\norm{U_j \dots U_1 \ket{\psi_1}-\ket{\psi_{j+1}}}
	\leq j \sqrt{\frac{2m}{1-12m\mu}} \left(2\sqrt{z} + \max_j\sqrt{ \kappa_j}\right)
	\leq 3j\sqrt{\frac{2m z}{1-12m\mu}}
	\leq \frac{6j\mu}{m},
	\label{distanceUjPsijppUp}
\end{align}
where the upper bound on $\kappa_j$ and $z$ comes from \eqref{tchoice}, 
and we assume $12m\mu \ll 1$.
For $j=m$, this also means the last claim of this Lemma holds: 
$\norm{U_m \dots U_1 \ket{\psi_1}-\ket{\psi_{m+1}}}	\leq 6\mu$.
\end{proof}

The guarantee \eqref{Abound} for the state $\ket{S}$ means we have probability at least $\frac{1}{2m}-6\mu$ to measure any $i$,
when measuring the label register. Thus, we can obtain any $\ket{\psi_i}$ with reasonable probability, and use it to check if the whole sequence in $\ket{S}$ is properly initialized and finalized (for $i=1$ and $i=m+1$, with the \textsc{start} and \textsc{end} test), or to verify that each state in it has a low energy (with the \textsc{low} test). We will do this in the following Sections.

\subsubsection{Initial state and final state tests}
\label{sec:init}

The role of tests 6 (\textsc{start}) and 7 (\textsc{end}) is to check if the sequence $\ket{S}$ (and $\ket{S'}$) is actually relevant to the problem -- that it connects to the two states we want to traverse between in the ground space of the GSCON problem Hamiltonian. 

First, we have the \textsc{start} test. Thanks to Lemma~\ref{lem:TS}, we know there is a probability at least $\frac{1}{2m}-6\mu$ to measure $i=1$ in the label (first) register of $\ket{S}$, giving us $\ket{\psi_1}$ in the data (second) register. When we successfully \textsc{swap} it with the initial state $\ket{\psi}$ from the GSCON instance, we get a guarantee on their closeness.
The \textsc{end} test works analogously, for the $i=m+1$ case, comparing $\ket{\psi_{m+1}}$ with $\ket{\phi}$. However, note that we put much more emphasis on the \textsc{start} test, as we can rely on perfect completeness for collaborating Merlins, while
the \textsc{end} test has some probability of false rejections even for good proofs, thanks to the $\eta_3$ limitation from the problem instance.

We illustrate the following argument in Figure~\ref{fig:statedistance}.
The second claim of Lemma~\ref{lem:TS} guarantees that $\ket{\psi_{m+1}}$ is close to $U_m \dots U_1 \ket{\psi_{1}}$. This, in turn, is close to $U_m\dots U_1 \ket{\psi}$, 
because $\ket{\psi_1}$ is close to $\ket{\psi}$.
Thus, when we measure $i=m+1$ in the label register of $\ket{S}$ and obtain $\ket{\psi_{m+1}}$ in the data register, we can \textsc{swap} test it with the final state $\ket{\phi}$
from the GSCON instance. Again, a high success rate implies closeness of these states.
Combining these results implies that $U_m \dots U_1\ket{\psi}$ is strictly closer than $\eta_4$ to the final GSCON state $\ket{\phi}$.
Let us prove this.

\begin{figure}[!h]
	\begin{center}
		\includegraphics{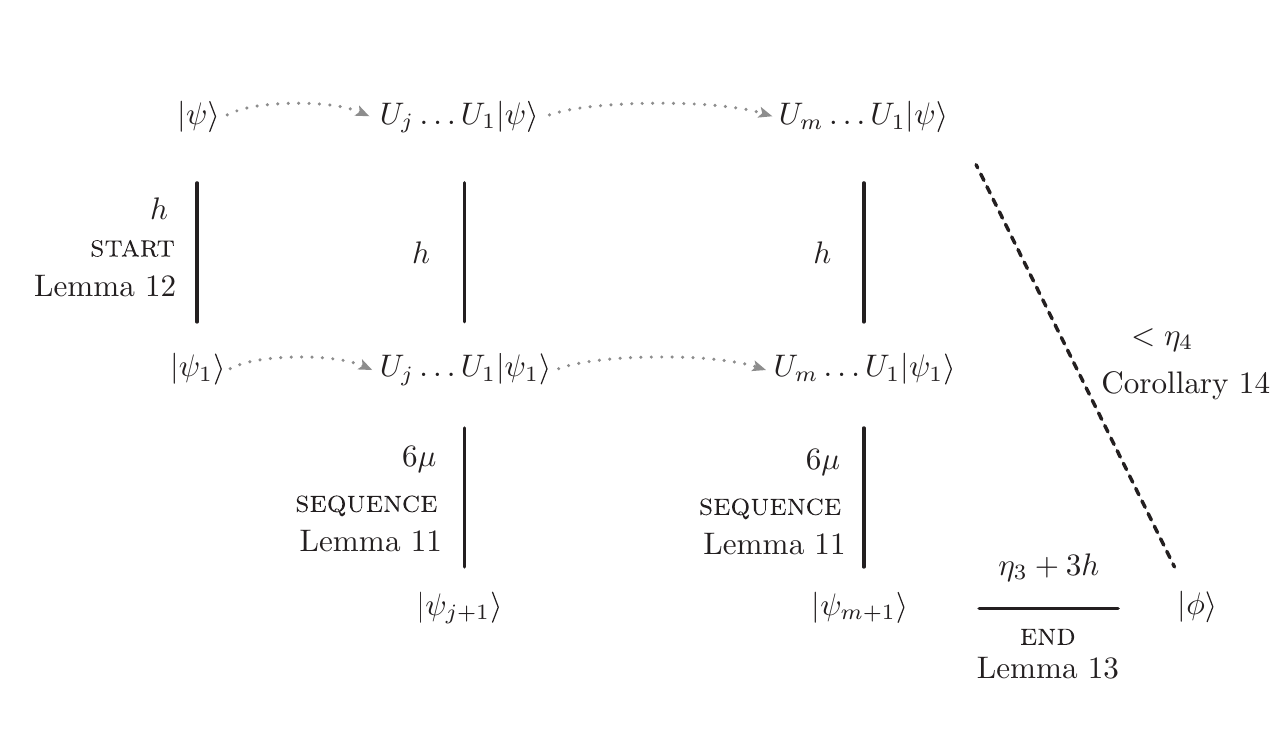}
		\caption{
			\label{fig:statedistance}
			Our goal is to understand the relationship of the state $\ket{\psi}$, its unitary transformations, and $\ket{\phi}$. However, the states that we work with are the $\ket{\psi_j}$'s. Passing tests 1-7 with high probability
			gives us upper bounds on the distance between the states (the black lines). The dashed line is an implication about the maximum distance of $\ket{\phi}$ and $U_m \dots U_1\ket{\psi}$, as 
			$h \leq \frac{1}{4}(\eta_4-\eta_3)$ \eqref{hset} and 
			$6\mu \leq h$ \eqref{muset}.
		Note also that the distance between $\ket{\psi_{j+1}}$ and $U_j \dots U_1 \ket{\psi}$ for all $j$ is not larger than $2h \leq \frac{1}{3}\sqrt{\eta_2/r}$ \eqref{hset}, which will be required for Corollary~\ref{lem:core}.
		}
	\end{center}
\end{figure}

Our goal is to set test 6 (\textsc{start}) up so that if it fails with probability below $p_6 r_6 = 1-s'$, we get a very good guarantee on the closeness of
$\ket{\psi_1}$ and $\ket{\psi}$. We then set test 7 (\textsc{end}) up so that if it fails with probability below $p_7 r_7 = 1-s'$, we get a strong guarantee on the closeness of
$\ket{\psi_{m+1}}$ and $\ket{\phi}$. Combined with the
result on the closeness of $\ket{\psi_{m+1}}$ and $U_m \dots U_1\ket{\psi_{1}}$,
we will thus arrive at a bound on the closeness of $\ket{\phi}$ and $U_m \dots U_1\ket{\psi}$, required to invoke the promise of the GSCON instance. 

On the other hand, in the completeness case, good proofs are rejected with probability at most $1-c' = \frac{p_7}{2m} \left(\frac{\eta_3^2}{2}-\frac{\eta_3^4}{8}\right)$, and we need to make sure that $c'$ is at least an inverse-polynomial in $m$ larger than the soundness bound $s'$, leaving open a completeness-soundness gap.

Before we turn to the tests in more detail, let us look at the \textsc{swap} test one last time to discuss a technical issue -- freedom of phase.
A \textsc{swap} test on $\ket{a}$ and $\ket{b}$ passes with the same probability 
as a \textsc{swap} test on $\ket{a}$ and $\ket{b'}=e^{i\omega}\ket{b}$ for some phase $e^{i\omega}$. 
There exists a phase $e^{i\omega}$such that $\braket{a}{b'}$ is real and nonnegative. Thus, the rejection probability of the \textsc{swap} test is
\begin{align}
	\frac{1-|\braket{a}{b}|^2}{2} = \frac{1-|\braket{a}{b'}|^2}{2} = \frac{1-\left(1-\frac{1}{2}\norm{\ket{a}-\ket{b'}}^2\right)^2}{2} = \frac{w^2}{2}-\frac{w^4}{8},
	\label{pfailswap}
\end{align}
where $w=\norm{\ket{a}-\ket{b'}}$,
as for real and nonnegative $\braket{a}{b'}$ we can write
$\norm{\ket{a}-\ket{b'}}^2 = 2 - 2 \textrm{Re}(\braket{a}{b'})
= 2-2|\braket{a}{b'}|=w$. Note that the maximum value of $w$ is $\sqrt{2}$, when we look at two orthogonal states. We also know that \eqref{pfailswap} is a growing function of $w$ for $0\leq w \leq \sqrt{2}$, as the derivative of \eqref{pfailswap} is $w\left(1-\frac{w^2}{2}\right)$. 

Let us consider test 6 (\textsc{start}).
When the previous tests pass with high enough probability, in the {\em NO} case, the probability to measure $i=1$ in the label register of $\ket{S}$ is at least $\frac{1}{2m}-6\mu$. 
We then perform a \textsc{swap} test between $\ket{\psi_1}$ and $\ket{\psi}$. 
For $\norm{\ket{\psi_1}-\ket{\psi}} = w$ (w.l.o.g. assuming real and nonnegative $\braket{\psi_1}{\psi}$), this test fails with probability 
$\frac{w^2}{2}-\frac{w^4}{8}$.
The provers' best shot at tricking the verifier is to maximize $w$, while keeping the overall probability of test failure below what is asked for in the test.
Formally:

\begin{lemma}[\textsc{start} test soundness]
\label{lem:distfailstart}
	Assume the previous tests would fail with respective probabilities below $r_1,\dots, r_5$ from Figure~\ref{tab:prob}. 
	If test 6 (\textsc{start}) would fail with probability below $r_6 = \left(\frac{1}{2m}-6\mu\right)\frac{h^2}{4}$
	for some $h\leq\sqrt{2}$,
	then 
	the states $\ket{\psi_1}$ and $\ket{\psi}$ are close, i.e. there exists a phase $e^{i\omega_{\psi}}$ such that 
	$\norm{\ket{\psi_1}-e^{i\omega_{\psi}}\ket{\psi}} < h$.
\end{lemma}
\begin{proof}	
	As the previous tests would pass with high enough probability,
	Lemma~\ref{lem:TS} guarantees that the probability of measuring $i=1$ in the label register is at least $\frac{1}{2m}-6\mu$. 
	Let us calculate the failure probability of the \textsc{start} test, using the \textsc{swap} test rejection probability \eqref{pfailswap}:
	\begin{align}
		\left(\frac{1}{2m}-6\mu\right)\left(\frac{w^2}{2}-\frac{w^4}{8}\right) \geq \left(\frac{1}{2m}-6\mu\right)\frac{w^2}{4},
	\end{align}
	because $\frac{w^2}{2}-\frac{w^4}{8}\geq\frac{w^2}{4}$ for $w\leq \sqrt{2}$.
	Thus, there must exist a $\omega_{\psi}$ such that 
	\begin{align}
		\norm{\ket{\psi_1}-e^{i\omega_{\psi}}\ket{\psi}} = w < h,
	\end{align}
	in order that the failure probability remains below
	$r_6 = \left(\frac{1}{2m}-6\mu\right)\frac{h^2}{4}$, which we assumed in the Lemma.
\end{proof}

\begin{lemma}[\textsc{end} test soundness]
\label{lem:distfailend}
	Assume tests 1-6 would fail with respective probabilities below $r_1, \dots, r_6$ listed in Figure~\ref{tab:prob}. If test 7 (\textsc{end}) would reject with probability below $r_7 = 
	\left(\frac{1}{2m} - 6\mu\right)
	\left(\frac{(\eta_3+h)^2}{2}-\frac{(\eta_3+h)^4}{8}\right)$,
	such that $\eta_3+h \leq \sqrt{2}$,
	then there exists a phase $e^{i\omega_{\phi}}$ such that $\norm{\ket{\psi_{m+1}}-e^{i\omega_{\phi}}\ket{\phi}} < \eta_3+h$.
\end{lemma}
\begin{proof}	
	As the previous tests pass with high enough probability,
	Lemma~\ref{lem:TS} guarantees that the probability of measuring $i=m+1$ in the label register is at least $\frac{1}{2m}-6\mu$. 
	We observe that if there did not exist a phase $\omega_{\phi}$ such that
	$\norm{\ket{\psi_{m+1}}-e^{i\omega_{\phi}}\ket{\phi}} =w < \eta_3+h$,
	the rejection probability of the \textsc{swap} test \eqref{pfailswap}
	between $\ket{\psi_{m+1}}$ and $\ket{\phi}$
	would be $\frac{w^2}{2}-\frac{w^4}{8}\geq 
	\frac{(\eta_3+h)^2}{2}-\frac{(\eta_3+h)^4}{8}$,
	as this is a growing function of $w$ for $0\leq w\leq \sqrt{2}$.
	Thus, the rejection probability of the \textsc{end} test would be at least $r_7$, a contradiction. Therefore, the claim is true.
\end{proof}

We will choose $h$ as the minimum of two values calculated from the GSCON instance parameters, recalling that $\Delta \leq \min \left\{\eta_4-\eta_3,  \eta_2\right\}$:
\begin{align}
	h = \min \left\{
		\frac{\eta_4-\eta_3}{4}, \frac{1}{6}\sqrt{\frac{\eta_2}{R}}
	\right\}
	 = \min \left\{
		\frac{\Delta}{4}, \frac{1}{6}\sqrt{\frac{\Delta}{R}}
	\right\}.
	\label{hset}
\end{align}
The first value implies 
	$\eta_3+h < \eta_4 \leq \sqrt{2}$ required for Lemma~\ref{lem:distfailend}
	and Corollary~\ref{cor:final}. The second value is required later in Corollary~\ref{lem:core}.
	We then choose 
\begin{align}
	\mu = \frac{h^2}{144m(\eta_3+h)},
	\label{muset}
\end{align}
so that $6\mu \leq h$ (required for Figure~\ref{tab:prob} and Corollary~\ref{cor:final}), $\mu<\frac{1}{36m}$ (required in the proof of Lemma~\ref{lem:lowe}), as well as for a gap result (required in the completeness Section~\ref{sec:completeness}). 
Let us we define a parameter $\gamma$ by
\begin{align}
	\left(\frac{1}{2m} - 6\mu\right)
	\left(\frac{(\eta_3+h)^2}{2}-\frac{(\eta_3+h)^4}{8}\right) =
	\frac{1}{2m} \left(\frac{\eta_3^2}{2}-\frac{\eta_3^4}{8}\right) + \gamma.
\end{align}
We can now prove that
\begin{align}
	\gamma 
	&= \left(\frac{1}{2m} - 6\mu\right)\left(\frac{(\eta_3+h)^2}{2}-\frac{(\eta_3+h)^4}{8}\right)
		- \frac{1}{2m} \left(\frac{\eta_3^2}{2}-\frac{\eta_3^4}{8}\right) \\
	&= \frac{1}{2m} \left(\frac{(\eta_3+h)^2-\eta_3^2}{2}-\frac{(\eta_3+h)^4-\eta_3^4}{8}
		-12\mu m\left(\frac{(\eta_3+h)^2}{2}-\frac{(\eta_3+h)^4}{8}\right)
		\right)\\
	&\geq 	
		\frac{1}{2m} \left(\frac{h(2\eta_3+h)}{2}
			\frac{h}{4}\left(2\sqrt{2}-h\right)
			-6\mu m(\eta_3+h)^2 \right)\label{eq:problem}\\
		&\geq 	
		\frac{h^2(\eta_3+h)}{16m} \left(\sqrt{2} -\frac{48\mu m (\eta_3+h)}{h^2}\right)
	\geq 	
		\frac{h^2(\eta_3+h)}{16m}, 
		\label{gammabound}
\end{align}
labeling $a=\eta_3+h$ and $b=\eta_3$ and utilizing
$a^2-b^2 -\frac{1}{4}(a^4-b^4) = (a-b)(a+b)
\left(1-\frac{a^2+b^2}{4}\right)$,
which is not larger than $\frac{1}{4}(a-b)(a+b)h\left(2\sqrt{2}-h\right)$,
when we realize that $a\leq \sqrt{2}$ and $b\leq \sqrt{2}-h$.

With this in hand, we can show that if the Tests 1-7 pass with high enough probability, the state $\ket{S}$ contains enough information about the state sequence $U_j \dots U_1 \ket{\psi}$, with $U_j$'s from the state $\ket{U}$. 
In other words, the states $U_m \dots U_1 \ket{\psi}$ and $\ket{\phi}$ must be close up to a phase.

\begin{corollary}[GSCON final state condition] Assume tests 1-7 would fail with respective probabilities below $r_1,\dots, r_7$, given in Figure~\ref{tab:prob}. Then 
there is a phase $e^{i\omega_{\phi}}$ such that 
$\norm{U_m\ldots U_1\ket{\psi}-e^{i\omega_\phi}\ket{\phi}}< \eta_3 + 3h 
< \eta_4$.
\label{cor:final}
\end{corollary}
\begin{proof}
To show this, we combine the previous results, as illustrated in Figure~\ref{fig:statedistance},
and recall that $h \leq \frac{\eta_4-\eta_3}{4}$ \eqref{hset}.
\begin{enumerate}
\item $\ket{\psi_{j+1}}$ is close to $U_j \dots U_1 \ket{\psi_1}$, thanks to Lemma~\ref{lem:TS} about the \textsc{sequence} test.
\item $\ket{\psi_1}$ is close to $\ket{\psi}$,
 thanks to Lemma~\ref{lem:distfailstart} about the \textsc{start} test.
\item $\ket{\psi_{m+1}}$ is close to $\ket{\phi}$, thanks to Lemma~\ref{lem:distfailend} about the \textsc{end} test.
\end{enumerate}
Combining the previous two results, we realize that $\ket{\psi_{m+1}}$ has high overlap with $U_m \dots U_1 \ket{\psi}$.
In detail, the triangle inequality tells us
that without loss of generality, we can choose the phases of the vectors so that the overlaps are real and nonnegative, and 
\begin{align}
	\norm{U_m\ldots U_1\ket{\psi}-e^{i\omega_{\phi}}\ket{\phi}} &\leq 
		  \norm{U_m\ldots U_1\ket{\psi_1}-\ket{\psi_{m+1}}} \nonumber\\
		&+ \norm{U_m\ldots U_1\ket{\psi}-U_m\ldots U_1\ket{\psi_1}} 
		+ \norm{\ket{\psi_{m+1}}-\ket{\phi}} \nonumber\\
	&< 6\mu + h + (\eta_3+h)  \leq \eta_3 + 3h <  \eta_4,
	\label{guaranteephi}
\end{align}
as guaranteed by the results 1-3 described above and our choice of $\mu$ \eqref{muset}.
\end{proof}

Therefore, we now either have one of the tests 1-7 rejecting with probability at least $r_i$, resulting in overall acceptance at most $s' = 1 - p_i r_i$, 
or a guarantee that the state $\ket{S}$ is very close to a sequence of states
$\ket{j}\ket{\psi_j}$ with $\ket{\psi_j}= U_j \dots U_1 \ket{\psi}$ and $\ket{\psi_{m+1}}=\ket{\phi}$.
However, in the {\em NO} case this is impossible -- so the last test \textsc{low} should reject the proof. We show this in the next Section.

\subsubsection{Low energy testing}
\label{sec:energy}

We run the final test 8 (\textsc{low}) with probability $p_8$.
We will show that if it would pass with probability $\geq s'$, it would mean the states in the sequence $U_j \dots U_1 \ket{\psi}$ 
have energy strictly below $\eta_2$. 
However, thanks to the promise of the GSCON problem, in the {\em NO} case there doesn't exist a sequence of states $U_j \dots U_1 \ket{\psi}$ ending $< \eta_4$ close to $\ket{\phi}$, with all states with energy strictly below $\eta_2$.
This will mean that either the final test rejects with probability at least $r_8$, 
or one of the previous tests must reject with probability at least its $r_i$ (see Figure~\ref{tab:prob}).

We have chosen to analyze the frustration-free variant of GSCON, with positive semidefinite Hamiltonians and $\eta_1=0$, i.e. exactly traversing a frustration-free ground space, because we want to avoid a technical\footnote{
Our proof would also go through for a very small $\eta_1$, or could be avoided with more copies of the proof. However, we haven't found a good enough way of 
perfoming a single measurement of energy for non-frustration-free Hamiltonians,
that would with high enough probability tell if an energy of a single copy of a state (a superposition of eigenstates with various energies) is below or above thresholds that could be very close together.} issue. Since we only have one (two) copies of the witness, it is difficult to perform a precise enough low-energy test, which would not disturb the completeness of the procedure. 

In practice, we measure the label register of $\ket{S}$, obtaining a label $i$. Thanks to \eqref{bigAA}, we know the probability of measuring any is not too small. 
We then measure the energy of the state $\ket{\psi_i}$, and reject or accept depending on the result.
With the guarantees collected so far, we now have a sequence of states 
$U_j \dots U_1 \ket{\psi}$ that ends strictly closer than $\eta_4$ to $\ket{\phi}$ \eqref{guaranteephi}.
Therefore, if we would test the energy of the states $U_j \dots U_1 \ket{\psi}$, at least one state in the sequence must have energy above $\eta_2$.
Now, because the states $U_j \dots U_1 \ket{\psi}$ are close to the states $\ket{\psi_{j+1}}$, testing whether the energy of the $\ket{\psi_{j+1}}$'s 
is low allows us to test if the energy of the $U_j \dots U_1 \ket{\psi}$'s
is low enough.
In detail,
\begin{lemma}[Low energy testing, frustration-free case]
\label{lem:lowe}
Assume tests 1-7 
would fail with respective probabilities below $r_1,\dots, r_7$, listed in Figure~\ref{tab:prob}.
If the final test (\textsc{low}) fails with probability below 
$r_8 = \frac{\eta_2}{8Rm}$,
with $R$ the number of terms in the Hamiltonian of the ff-GSCON instance, then the energy of each state $\ket{\psi_{i}}$
must be below $\frac{\eta_2}{2}$.
\end{lemma}
\begin{proof}

Assuming that tests 1-7 
pass with the probabilities denoted in Figure~\ref{tab:prob}, Lemma~\ref{lem:TS} guarantees that the probability of measuring any $i$ in the label register of state $\ket{S}$ is at least $\frac{1}{2m}-6\mu$. 
When we measure the label register of the state $\ket{S}$, we obtain some value $i$ and a state $\ket{\psi_i}$ in the data register.
Using a measurement circuit \cite[p.142-143]{kitaev2002classical} for a 
local Hamiltonian, 
we can now measure the energy of $\ket{\psi_i}$ for our GSCON Hamiltonian $H$ 
made from $r$ positive semidefinite terms with norm at most 1.
We will reject if this circuit outputs 0, which happens with probability $\frac{1}{R}\bra{\psi_i}H\ket{\psi_i}$.

Now, assume the energy of a state $\ket{\psi_i}$ was above $\frac{\eta_2}{2}$,
the rejection probability for the energy measurement circuit 
would be above $\frac{\eta_2}{2R}$.
The rejection probability of the \textsc{low} test would thus be above 
\begin{align}
	\left(\frac{1}{2m}-6\mu\right) \frac{\eta_2}{2R}
	> \frac{\eta_2}{8Rm}, \label{lowfail}
\end{align}
as $\mu < \frac{1}{24m}$ thanks to \eqref{muset}. 
However, this \eqref{lowfail} disagrees with the assumption of the Lemma.
Therefore, the energy of each $\ket{\psi_i}$ must not be above $\frac{\eta_2}{2}$.
\end{proof}

We can now finally show that if test 8 passes with high probability, the energy of each state $U_j \dots U_1 \ket{\psi}$ must be low enough.
\begin{corollary}[Low energy requirement for GSCON]
Assume tests 1-8 would fail with respective probabilities below $r_1, \dots, r_8$, listed in Figure~\ref{tab:prob}.
Then the energy of each state $U_j \dots U_1 \ket{\psi}$ is strictly below $\eta_2$.
\label{lem:core}
\end{corollary}
\begin{proof}
We already know \eqref{lowfail}. Thanks to Lemma~\ref{lem:TS} about the \textsc{sequence} test and Lemma~\ref{lem:distfailstart} about the \textsc{start} test,
we also know that
\begin{align}
	\Delta_{\psi_{j+1}} = \norm{\ket{\Delta_{\psi_{j+1}}}} \leq 2h,
\end{align}
when we label $\ket{\Delta_{\psi_{j+1}}}=U_j \dots U_1 \ket{\psi} - \ket{\psi_{j+1}}$.
Therefore, if Lemma~\ref{lem:lowe} says the energy of $\ket{\psi_{j+1}}$ is at most $\frac{\eta_2}{2}$,
the energy of the state $U_j \dots U_1 \ket{\psi}$ is
\begin{align}
	\bra{\psi} U_j^\dagger \dots U_1^\dagger  H U_j \dots U_1 \ket{\psi}
	&=
	\bra{\psi_{j+1}} H \ket{\psi_{j+1}}
	+ \bra{\Delta_{\psi_{j+1}}} H \ket{\psi_{j+1}}
	+ \bra{\psi_{j+1}} H \ket{\Delta_{\psi_{j+1}}}
	+ \bra{\Delta_{\psi_{j+1}}} H \ket{\Delta_{\psi_{j+1}}} \nonumber \\
	&\leq \frac{\eta_2}{2} \left( 1+ 2\Delta_{\psi_{j+1}} \right) + \opnorm{H} \Delta_{\psi_{j+1}}^2
	\leq \frac{\eta_2}{2} \left( 1+ 2\Delta_{\psi_{j+1}} \right) + R \Delta_{\psi_{j+1}}^2.
\end{align}
Recall that we chose $h$ \eqref{hset} so that
$\Delta_{\psi_{j+1}} \leq 2h \leq \frac{1}{3}\sqrt{\frac{\eta_2}{R}}$,
which means $2\Delta_{\psi_{j+1}} \leq \frac{2}{3}$
and $R\Delta_{\psi_{j+1}}^2  \leq \frac{\eta_2}{9}$.
The energy of each $U_j \dots U_1 \ket{\psi}$ is thus upper bounded by 
$\frac{17}{18}\eta_2 < \eta_2$. This is strictly below $\eta_2$, as we wanted to show.
\end{proof}

Let us now combine the results we have proven so far. 
The argument from Lemma~\ref{lem:swapU} to Corollary \ref{lem:core} collectively says that if the tests pass with high enough probability,
we get guarantees about the sequence $U_j \dots U_1 \ket{\psi}$.
However, all those guarantees are irreconcilable with the {\em NO} case of the GSCON problem -- there must be at least one state
$U_j \dots U_1 \ket{\psi}$
with energy at least $\eta_2$, if the sequence starts at $\ket{\psi}$ and ends near enough $\ket{\phi}$.
Therefore, at least one of the tests must fail with probability at least the respective rejection threshold $r_i$ listed in Figure~\ref{tab:prob}.
There, we also set the probabilities $p_1,\dots,p_8$ for running the tests, 
so that when we combine them with the desired thresholds for rejection $r_i$,
we get $p_i r_i = 1-s'$ for some $s'$. This is easily achievable for the following set of $p_i$'s:
\begin{align}
	p_i = \frac{r_i^{-1}}{\sum_j r_j^{-1}},
\end{align}
which obey $\sum_i p_i = 1$, with each $p_i$ at least an inverse polynomial in $m$. 
Thus, in the {\em NO} case, the probability to pass the tests for cheating provers is at most $s'=1-\left(\sum_j r_j^{-1}\right)^{-1}$ \eqref{sfromr}.

This concludes the {\em soundness} part of Theorem~\ref{th:main}.
It remains to show {\em completeness} -- the acceptance probability of the protocol in the {\em YES} case must be at least $c'$, which needs to be at least an inverse polynomial in $n$ above $s'$.

\subsection{Completeness}
\label{sec:completeness}

Let us run through how well the tests can run in the {\em YES} case, with honest provers following the protocol, and show a high probability of acceptance if the provers behave honestly.

The states $\ket{U}$ and $\ket{U'}$ are identical, and contain computational-basis encoded unitary gates $U_i$. The first test, \textsc{swap~U}, 
and the second test, \textsc{unique}, thus pass perfectly. 

Let us look at the third test, \textsc{uniform} on a state $\ket{U}$ with a proper form. The probability to pass the gate-register projection is exactly $\frac{1}{G}$. When this passes, the probability to pass the label-register projection is exactly 1. Therefore, this test also passes perfectly for honest provers.

We can turn to the tests for the $\ket{S}$ states. As before, Test 4 (\textsc{swap~S}) passes perfectly. What about Test 5 (\textsc{sequence})?
The probabilistic procedure that applies the gates in $\ket{U}$ to the state $\ket{S}$ works with probability at least $\frac{1}{2m}$. After the shift in the label register, the comparison with the state $\ket{S'}$ passes perfectly.
Altogether, Test 5 passes perfectly again.

The sixth test, \textsc{start}, passes perfectly on a proper witness that has $\ket{\psi_1}=\ket{\psi}$.

The seventh test, \textsc{end}, involves $i=m+1$, where we have a problem instance promise that there exists a state $\ket{\psi}$ for which 
$\ket{\psi_{m+1}}= U_m \dots U_1\ket{ \psi_1}
= U_m \dots U_1\ket{ \psi}$ is at most $\eta_3$ far from $\ket{\phi}$.
Therefore, when we check this with a \textsc{swap} test, there is a chance at most
$\frac{1}{2m}\left(\frac{\eta_3^2}{2}-\frac{\eta_3^4}{8}\right)$
to reject a good witness.

Finally, we chose to look at ff-GSCON instances with $\eta_1=0$, with a frustration-free, positive semidefinite Hamiltonian, all the states $\ket{\psi_j}$ have energy exactly zero for all of the Hamiltonian's terms, so the final test, \textsc{low}, passes perfectly.
If we did not choose this variant of GSCON, the ambiguity in the low-energy testing in the {\em YES} case could reduce the completeness unfavorably. We leave as an open question, whether this requirement can be removed or not.

Altogether, the probability to pass the whole procedure for honest, unentangled provers is at least
\begin{align}
	c' \geq 1-  \frac{p_7}{2m} \left(\frac{\eta_3^2}{2}-\frac{\eta_3^4}{8}\right), 
\end{align}
thanks to the freedom of the final state $\ket{\psi_{m+1}}$ to be a little bit away from the expected final state $\ket{\phi}$ in test 7.

\subsection{The completeness-soundness gap is an inverse polynomial}
\label{sec:csgap}

Let us now compare the completeness bound from Section \ref{sec:completeness} with the probability of acceptance in the {\em NO} case from Section~\ref{sec:soundness}.
We choose the respective test-running probabilities $p_1, \dots, p_8$ according to \eqref{pchoice} 
so that $p_i r_i = 1-s'$, where $r_i$ is the desired maximum rejection probability for a given test, listed in Figure~\ref{tab:prob}. 
It means the maximum acceptance probability in the {\em NO} case is $s'$.
Thanks to $p_7 r_7 = 1-s'$ and recalling \eqref{gammabound}, 
the completeness-soundness gap for our four-unentangled-state protocol thus obeys
\begin{align}
 c'-s' \geq 1-p_7 \left(\frac{\eta_3^2}{2}-\frac{\eta_3^4}{8}\right)- \left(1-p_7 r_7\right)
= p_7 \gamma \geq \frac{p_7 h^2(\eta_3+h)}{16m}.\label{csgaplowerbound}
\end{align}
It is at least an inverse polynomial in $m$ and thus $n$.
This concludes the proof of Theorem~\ref{th:main}.

For those wishing to see what a terrible inverse polynomial it is, let us make an estimate. First, we need a lower bound on $p_7=r_7^{-1}/\sum_{j}r_j^{-1}$.
Looking at Figure~\ref{tab:prob}, we see that the prohibitively dominating term in $\sum_j r_j^{-1}$ is $r_1^{-1}$, and we can upper bound it by $O\left(m^{32} G^{10} \Delta^{-12}\right)$,
where $\Delta$ is an upper bound on $h$ in \eqref{hset}, coming from the GSCON parameters. On the other hand, we can estimate $r_7^{-1}$ to be roughly $m \Delta^{-2}$. Plugging these estimates into \eqref{csgaplowerbound}, we conclude that a lower bound on the final completeness-soundness gap is in $\Omega(\Delta^{13} m^{-32} G^{-10})$.

\section*{Acknowledgements}
DN's research has received funding from the People Programme (Marie Curie Actions) EU's 7th Framework Programme under REA grant agreement No. 609427. This research has been further co-funded by the Slovak Academy of Sciences.
LC and DN were also supported by the Slovak Research and Development Agency grant QETWORK APVV-14-0878.
MS thanks the Alexander-von-Humboldt Foundation for support.

\bibliographystyle{plain}	
\bibliography{qma2proofs}

\begin{thebibliography}{10}

\bibitem{Beigi0810.5109}
Scott Aaronson, Salman Beigi, Andrew Drucker, Bill Fefferman, and Peter Shor.
\newblock The power of unentanglement.
\newblock {\em Theory of Computing}, 5(1):1--42, 2009.

\bibitem{Beigi2outof4SAT}
Salman Beigi.
\newblock {$\NP$} vs {$\QMA_{\log}(2)$}.
\newblock {\em Quantum Information {\&} Computation}, 10(1{\&}2):2, 2010.

\bibitem{BlierTapp}
Hugue Blier and Alain Tapp.
\newblock {A Quantum Characterization Of NP}.
\newblock {\em computational complexity}, 21(3):499--510, Sep 2012.

\bibitem{BCY}
F.~G.~S.~L. {Brand{\~a}o}, M.~{Christandl}, and J.~{Yard}.
\newblock {Faithful Squashed Entanglement}.
\newblock {\em Communications in Mathematical Physics}, 306:805--830, September
  2011.

\bibitem{QuantumFingerprinting}
Harry Buhrman, Richard Cleve, John Watrous, and Ronald de~Wolf.
\newblock Quantum fingerprinting.
\newblock {\em Phys. Rev. Lett.}, 87:167902, Sep 2001.

\bibitem{ChenDrucker}
J.~{Chen} and A.~{Drucker}.
\newblock {Short Multi-Prover Quantum Proofs for SAT without Entangled
  Measurements}.
\newblock {\em ArXiv e-print: 1011.0716}, November 2010.

\bibitem{cj13-01}
Alessandro Chiesa and Michael~A. Forbes.
\newblock {Improved Soundness for $\QMA$ with Multiple Provers}.
\newblock {\em Chicago Journal of Theoretical Computer Science}, 2013(1),
  January 2013.

\bibitem{Dinur:2007:PTG:1236457.1236459}
Irit Dinur.
\newblock {The PCP Theorem by Gap Amplification}.
\newblock {\em J. ACM}, 54(3), June 2007.

\bibitem{LGNN}
Fran\c{c}ois~Le Gall, Shota Nakagawa, and Harumichi Nishimura.
\newblock On qma protocols with two short quantum proofs.
\newblock {\em Quantum Info. Comput.}, 12(7-8):589--600, July 2012.

\bibitem{GharibianStrongNPHardness}
Sevag Gharibian.
\newblock {Strong NP-hardness of the Quantum Separability Problem}.
\newblock {\em Quantum Info. Comput.}, 10(3{\&}4):343--360, March 2010.

\bibitem{gscon}
Sevag Gharibian and Jamie Sikora.
\newblock {Ground State Connectivity of Local Hamiltonians}.
\newblock In {\em {Automata, Languages, and Programming: 42nd International
  Colloquium, ICALP 2015, Kyoto, Japan, July 6-10, 2015, Proceedings, Part I}},
  pages 617--628, Berlin, Heidelberg, 2015. Springer Berlin Heidelberg.

\bibitem{Gurvits}
Leonid Gurvits.
\newblock {Classical Deterministic Complexity of Edmonds' Problem and Quantum
  Entanglement}.
\newblock In {\em Proceedings of the Thirty-fifth Annual ACM Symposium on
  Theory of Computing}, STOC '03, pages 10--19, New York, NY, USA, 2003. ACM.

\bibitem{ProductTest}
A.~W. Harrow and A.~Montanaro.
\newblock An efficient test for product states with applications to quantum
  merlin-arthur games.
\newblock In {\em 2010 IEEE 51st Annual Symposium on Foundations of Computer
  Science}, pages 633--642, Oct 2010.

\bibitem{JordanNagajNishimura}
Stephen~P. Jordan, Hirotada Kobayashi, Daniel Nagaj, and Harumichi Nishimura.
\newblock Achieving perfect completeness in classical-witness quantum
  merlin-arthur proof systems.
\newblock {\em Quantum Info. Comput.}, 12(5-6):461--471, May 2012.

\bibitem{kitaev2002classical}
A.Y. Kitaev, A.~Shen, and M.N. Vyalyi.
\newblock {\em Classical and Quantum Computation}.
\newblock Graduate studies in mathematics. American Mathematical Society, 2002.

\bibitem{PhDLiu}
Y.-K. {Liu}.
\newblock {\em {The Complexity of the Consistency and N-representability
  Problems for Quantum States}}.
\newblock PhD thesis, University of California, San Diego, 2007.

\bibitem{BlierTappSoundnessGapImprovement}
{Shota Nakagawa, Harumichi Nishimura}.
\newblock {On the Soundness of the Blier-Tapp QMA Protocol}.
\newblock In {\em 23rd Quantum Information Technology Symposium (QIT23) (In
  Japanese)}, pages pp. 132--135. {[online:
  \url{http://www.math.cm.is.nagoya-u.ac.jp/~hnishimura/NN10.pdf}]}, 2010.

\end{thebibliography}

\end{document}